\documentclass[aps,pre,showpacs,twocolumn]{revtex4}
\usepackage{amsfonts}
\usepackage{amsmath}
\usepackage{amsthm}
\usepackage{graphics}
\usepackage{graphicx}
\usepackage{subfigure}
\usepackage{amssymb}
\usepackage{color}
\usepackage[T1]{fontenc}
\usepackage{amsmath,amssymb,color,epsfig,latexsym,natbib,graphicx,dcolumn}

\newcommand{\bomega}{\mbox{\boldmath{$\omega$}}}

\newcommand{\R}{\mathbb{R}}

\newcommand{\eref}[1]{{(\ref{#1})}}
\newcommand{\mbf}[1]{{\mathbf{#1}}}
\newcommand{\Li}[2]{{{\mathrm{Li}}\left(#1,#2\right)}}
\newcommand{\LiTyG}[2]{{{\widetilde{\mathrm{Li}}}\left(#1,#2\right)}}
\newtheorem{thm}{Theorem}
\newtheorem{cor}[thm]{Corollary}
\newtheorem{lem}[thm]{Lemma}
\newtheorem{defn}[thm]{Definition}
\newtheorem{hypo}[thm]{Hypothesis}

\begin{document}
\title{\bf Interplay between the Beale-Kato-Majda theorem and the analyticity-strip method to investigate numerically the incompressible Euler singularity problem} 

\author{Miguel D. Bustamante}
\affiliation{School of Mathematical Sciences,
University College Dublin,
Belfield, Dublin 4, Ireland}
\author{Marc Brachet}
\affiliation{Laboratoire de Physique Statistique de l'Ecole Normale
Sup{\'e}rieure, \\
associ{\'e} au CNRS et aux Universit{\'e}s Paris VI et VII,
24 Rue Lhomond, 75231 Paris, France.}
\date{\today}

\bigskip

\begin{abstract}
Numerical simulations of the incompressible Euler equations are performed using the Taylor-Green vortex initial conditions and resolutions up to $4096^3$.
The results are analyzed in terms of the classical analyticity strip method and
Beale, Kato and Majda (BKM) theorem.
A well-resolved acceleration of the time-decay of the width of the analyticity strip $\delta(t)$ is observed at the highest resolution for $3.7<t<3.85$ while preliminary $3D$ visualizations show the collision of vortex sheets.
The BKM criterium on the power-law growth of supremum of the vorticity, applied on the same time-interval, is not inconsistent with the occurrence of a singularity around $t \simeq 4$.

These new findings lead us to investigate how fast the analyticity strip width needs to decrease to zero
in order to sustain a finite-time singularity consistent with the BKM theorem.
A new simple bound of the supremum norm of vorticity in terms of the energy spectrum is introduced and used to combine the BKM theorem with the analyticity-strip method.
It is shown that a finite-time blowup can exist only if $\delta(t)$ vanishes sufficiently fast at the singularity time. In particular, if a power law is assumed for $\delta(t)$ then its exponent must be greater than some critical value, thus providing a new test that is applied to our $4096^3$ Taylor-Green numerical simulation.

Our main conclusion is that the numerical results are not inconsistent with a singularity but that higher-resolution studies are needed to extend the time-interval on which a well-resolved power-law behavior of $\delta(t)$ takes place, and check whether the new regime is genuine and not simply a crossover to a faster exponential decay.
\end{abstract}

\maketitle
\section{Introduction} \label{Sec:Intro}
A central open question in classical fluid dynamics is whether the incompressible three-dimensional Euler equations with smooth initial conditions develop a singularity after a finite time.
A key result was established in the late eighties by Beale, Kato and Majda (BKM). The BKM theorem \cite{BKM84} states that blowup (if it takes place) requires the time-integral of the supremum of the vorticity to become infinite (see the review by Bardos and Titi \cite{BardosTiti2007}).
Many studies have been performed using the BKM result to monitor the growth of the vorticity supremum in numerical simulations in order to conclude yes or no regarding the question of whether a finite-time singularity might develop. The answer is  somewhat mixed, see {\it e.g.} references \cite{Kerr2005,Houli2006,BustamanteKerr2008} and the recent review by Gibbon \cite{Gibbon2008}.
Other conditional theoretical results, going beyond the BKM theorem, were obtained in a pioneering paper by Constantin, Fefferman and Majda \cite{ConstFetMaj96}. They showed that the evolution of the direction of vorticity posed geometric constraints on potentially singular solutions for the 3D Euler equation  \cite{ConstFetMaj96}. This point of view was further developed by Deng, Hou and Yu  in references \cite{Denghouyu2005} and \cite{Denghouyu2006}.

An alternative way to extract insights on the singularity problem from numerical simulations is the so-called analyticity strip method \cite{SulemSulemFrisch1983}.
In this method the time is considered as a real variable and the space-coordinates are considered as complex variables.
The so-called ``width of the analyticity strip'' $\delta (\geq0)$ is defined as the imaginary part of the complex-space singularity of the velocity field nearest to the real space.
The idea is to monitor $\delta(t)$ as a function of time $t$.
This method uses the rigorous result \cite{BardosBenachour77} that a real-space singularity of the Euler equations occurring at time $T_*$ must be preceded by a non-zero $\delta(t)$ that vanishes at $T_*$. Using spectral methods \cite{Got-Ors}, $\delta(t)$ is obtained directly from the high-wavenumber exponential fall off of the spatial Fourier transform of the solution \cite{frischbook}. This method effectively provides a ``distance to the singularity'' given by $\delta(t)$ \cite{Frisch:notout}, which cannot be obtained from the general BKM theorem.

Note that the BKM theorem is more robust than the analyticity-strip method in the sense that it applies to velocity fields that do not need to be analytic. However, in the present paper we will concentrate on initial conditions that are analytic. In this case, there is a well-known result that states:
\emph{In three dimensions with periodic boundary
conditions and analytic initial conditions, analyticity is preserved
as long as the velocity is continuously differentiable}\/ ($C^1$)
\emph{in the real domain}\/ \cite{BardosBenachour77}.
The BKM theorem allows for a strengthening of this result: analyticity is actually preserved as long as the vorticity is finite \cite{Frisch:notout}.

The analyticity-strip method has been applied to probe the Euler singularity problem
using a standard periodic (and analytical) initial data: the so-called
Taylor-Green (TG) Vortex \cite{TG1937}.
We now give a short review of what is already known about the TG dynamics.
Numerical simulations of the TG flow were performed with resolution increasing over the years, as more computing power became available.
It was found that except for very short times and for as long as $\delta(t)$ can be reliably measured, it displays almost perfect exponential decrease. Simulations performed in $1982$ on a grid of $256^3$ points obtained $\delta(t)\sim 2.60 \, e^{-t/0.57}$ (for $t$ up to $2.5$) \cite{BRACHET:1983p4817}. This behavior was confirmed in $1992$ at resolution $864^3$ \cite{Brachet1992}. More than $20$ years after the first study, simulations performed on a grid of $2048^3$ points yielded $\delta(t)\sim 2.70 \,  e^{-t/0.56}$ (for $t$ up to $3.7$) \cite{CichowlasBrachet2005}.
If these results could be safely extrapolated to later times then the Taylor-Green vortex would never develop a real singularity \cite{frischbook}.

The present paper has two main goals. One is to report on and analyze new simulations of the TG vortex that are performed at resolution $4096^3$. These new simulations show, for the first time, a well-resolved change of regime, leading to a faster decay of $\delta(t)$ happening at a time where preliminary $3D$ visualizations show the collision of vortex sheets \footnote{This new behavior of the Euler TG vortex is somewhat similar to the acceleration in the decrease of $\delta$ that was reported in MHD for the so-called IMTG initial data at resolution $2048^3$ in reference \cite{PhysRevE.78.066401}.}.
The second goal of this paper is to answer the following question, motivated by the new behavior of the TG vortex: how fast does the analyticity-strip width have to decrease to zero in order to sustain a finite-time singularity, consistent with the BKM theorem?
To the best of our knowledge, this question has not been formulated previously.

To answer this question we introduce a new bound of the supremum norm of vorticity in terms of the energy spectrum. We then use this bound to combine the BKM theorem with the analyticity-strip method.
This new bound is sharper than usual bounds.
We show that a finite-time blowup exists only if the analyticity-strip width goes to zero sufficiently fast at the singularity time.
If a power-law behavior is assumed for $\delta(t)$ then its exponent must be greater than some critical value.
In other words, we provide a powerful test that can potentially rule out the existence of a finite-time singularity in a given numerical solution of Euler equations.
We apply this test to the data from the latest $4096^3$ Taylor-Green numerical simulation in order to see if the change of behavior in $\delta(t)$ can be consistent with a singularity.

The paper is organized as follows: Section \ref{Sec:Theo} is devoted to the basic definitions, symmetries and numerical method related to the inviscid Taylor-Green vortex.

In Sec. \ref{Sec:Numerics_Classical}, the new high-resolution Taylor-Green results are presented and are analyzed classically in terms of analyticity-strip method and BKM.

In Sec. \ref{Sec:As_Bkm}, the analyticity-strip method and BKM Theorem are bridged together. The section starts with heuristic arguments that are next formalized in a mathematical framework of definitions, hypotheses and theorems.

In Sec. \ref{Sec:NewAnal}, our new theoretical results are used to analyze the behavior of the decrement.

Section  \ref{Sec:Conclusion} is our conclusion.

The generalization to non TG-symmetric periodic flows of the results presented in Sec. \ref{Sec:As_Bkm} are described in an appendix.

\section{Definition of the system\label{Sec:Theo} }

\subsection{Basic definitions}

Let us consider the 3D incompressible Euler equations for the velocity field
$\mbf{u}(x,y,z,t) \in \R^3$ defined for $(x,y,z) \in \R^3$ and in a time interval $t \in [0,T)$:
\begin{eqnarray}
\label{eq:Euler}
 \frac{\partial \mbf{u}}{\partial t} + \mbf{u} \cdot \nabla \mbf{u} = - \nabla p\,, \qquad \quad \nabla \cdot \mbf{u} = 0.
\end{eqnarray}

The Taylor-Green (TG) flow \cite{TG1937} is defined by the $2 \pi$-periodic initial data $\mbf{u}(x,y,z,0)={\bf u}^{\mathrm{TG}}(x,y,z)$, where
\begin{equation*}
{\bf u}^{\mathrm{TG}}=(\sin(x) \cos(y) \cos(z),-\cos(x) \sin(y)\cos(z), 0).
\end{equation*}

The periodicity of ${\bf u}$ allows us to
define the (standard) Fourier representation
\begin{eqnarray}
 \widehat{\mbf{u}}(\mbf{k},t) &=& \frac{1}{(2 \pi)^3} \int_D {\bf u}({\bf x},t) \exp(-i \mbf{k}{\bf x}) d^3x \\
{\mbf{u}}(\bf{x},t) &=& \sum\limits_{\mbf{k}\in \mathbb{Z}^3}  \widehat{\mbf{u}}(\mbf{k},t) \exp(i \mbf{k}{\bf x})  \label{eq:Four2},
\end{eqnarray}

The kinetic energy spectrum $E(k,t)$ is defined as the sum over spherical shells
\begin{equation}
\label{eq:spectrum}
E(k,t) = \frac{1}{2} {\displaystyle \sum_{\mbf{k} \in \mathbb{Z}^3 \atop k-1/2  < |\mbf{k}| < k+1/2 }} |\widehat{\mbf{u}}({\bf k},t)|^2,
\end{equation}
and the total energy
\begin{equation*}
E =\frac{1}{2 (2 \pi)^3} \int_D {\left|{\bf u}({\bf x},t)\right|^2} d^3x = \frac{1}{2} {\displaystyle \sum_{\mbf{k} \in \mathbb{Z}^3}} |\widehat{\mbf{u}}({\bf k},t)|^2 ,
\end{equation*}
is independent of time because ${\bf u}$ satisfies the 3D Euler equations (\ref{eq:Euler}).

\subsection{Symmetries}
\label{subsec:symm}
A number of the symmetries of ${\bf u}^{\mathrm{TG}}$ are compatible with the equation of motions. They are, first,
rotational symmetries of angle $\pi$ around the axis $(x=z=\pi/2)$ and $(x=z=\pi/2)$; and of angle $\pi/2$ around the axis $(x=y=\pi/2)$.
A second set of symmetries corresponds to planes of mirror symmetry: $x=0,\pi$, $y=0,\pi$ and $z=0,\pi$. On the symmetry planes, the velocity ${\bf u}^{\mathrm{TG}}$ and the vorticity ${\bomega^\mathrm{TG}}={\bf\nabla} \times {{\bf u}^{\mathrm{TG}}}$ are (respectively) parallel and perpendicular to these planes that form the sides of the so-called impermeable box which confines the flow.

It is demonstrated in reference \cite{BRACHET:1983p4817}
that these symmetries imply that the Fourier expansions coefficients of the velocity field in Eq. \eqref{eq:Four2}
$\widehat{\mbf{u}}(m,n,p,t)$ vanishes unless $m,n,p$ are either all even
or all odd integers. This fact can be used in a standard way \cite{BRACHET:1983p4817} to reduce
memory storage and speed up computations.

\subsection{Numerical method}

The Euler Equations \eqref{eq:Euler} are solved numerically
using standard \citep{Got-Ors} pseudo-spectral methods
with resolution $N$.
Time marching is done with a second-order
Runge-Kutta finite-difference scheme.
The solutions are dealiased by suppressing, at each time step, the
modes for which at least one wave-vector component exceeds
two-thirds of the maximum wave-number $N/2$
(thus a $4096^3$ run is truncated at $k>k_{\max} \equiv 1365$).

The simulations reported in this paper were performed using a special purpose symmetric parallel code developed from that described in \cite{PhysRevE.78.066401,PLBMR2010}.
The workload for a timestep is (roughly) twice that of a general periodic code running at a quarter of the resolution.
Specifically, at a given computational cost, the ratio of
the largest to the smallest scale available to a computation
with enforced Taylor-Green symmetries is enhanced by a
factor of $4$ in linear resolution. This leads to a factor of $32$ savings in total computational
time and memory usage.
The code is based on FFTW and a hybrid MPI-OpenMP scheme derived from that described in \cite{Mininni:GHOST}.
The runs were performed on the IDRIS BlueGene/P machine.  At resolution $4096^3$ we used $512$ MPI processes,
each process spawning $4$ OpenMP threads.

\section{Numerical results and classical analysis}
\label{Sec:Numerics_Classical}
\begin{figure}[htbp]
\begin{center}
\includegraphics[height=10.0cm]{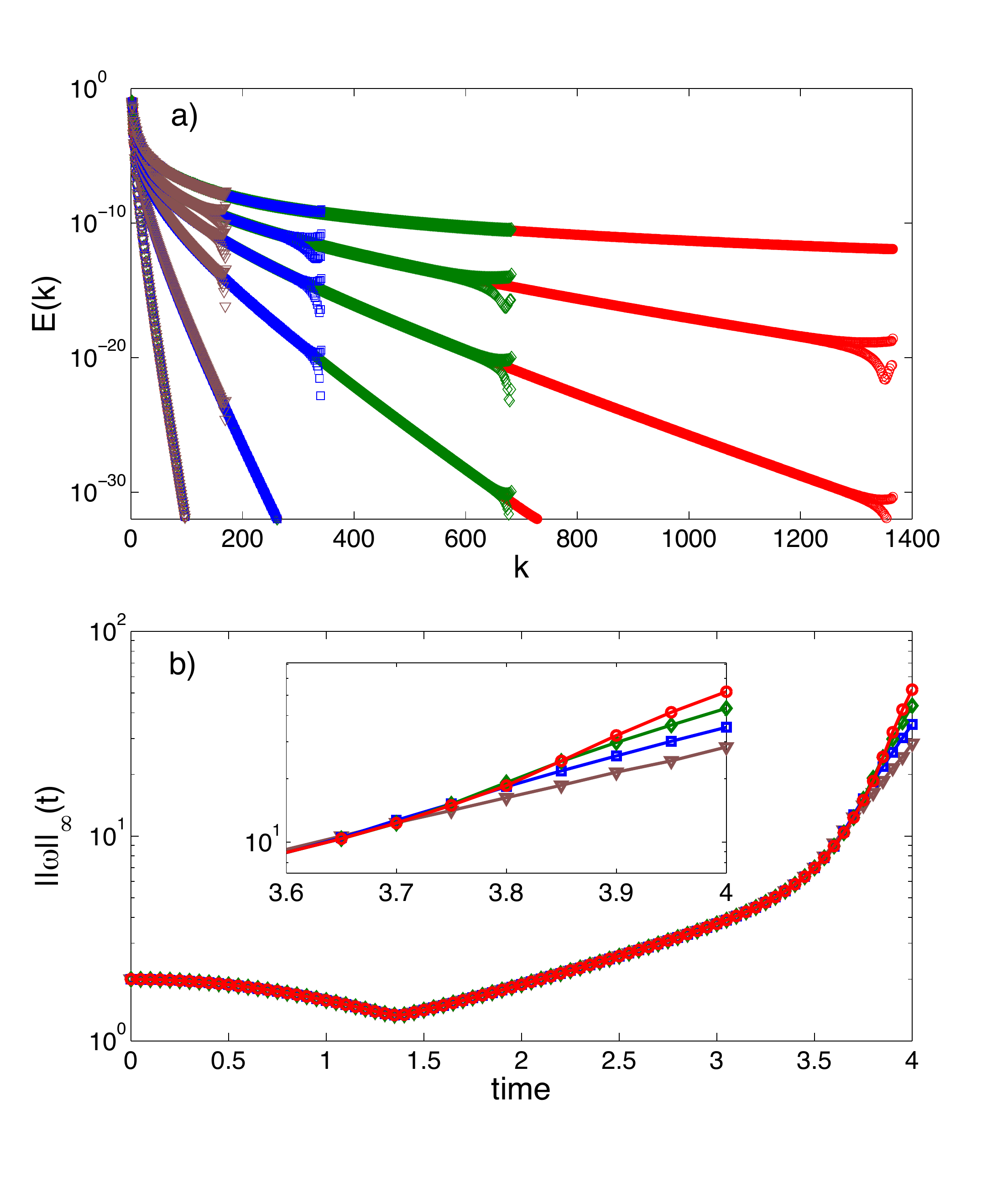}
\caption{(Color online) Temporal evolution of TG flow: a) energy spectra $E(k,t)$ (see Eq. \eqref{eq:spectrum}) at $t=(1.3, 1.9, 2.5, 2.9, 3.4, 4.0)$ and  b) maximum of vorticity $\left\| \bomega(\cdot,t) \right\|_{\infty}$. Results from runs performed at different resolutions are displayed together: $512^3$ (brown triangles), $1024^3$ (blue squares), $2048^3$ (green diamonds) and $4096^3$ (red circles).}
\label{Fig:Energy_Spectra_maxvort}
\end{center}
\end{figure}

\subsection{Energy spectra, maximum vorticity and collision of vortex sheets}
Runs were performed at resolutions $512^3$, $1024^3$, $2048^3$ and $4096^2$.

The behavior of the energy spectra \eqref{eq:spectrum} and the spatial maximum of the norm of the vorticity $\bomega=\nabla \times \mbf{u}$ are presented in Fig. \ref{Fig:Energy_Spectra_maxvort}.

\begin{figure}[htbp]
\begin{center}
\includegraphics[height=9.5cm]{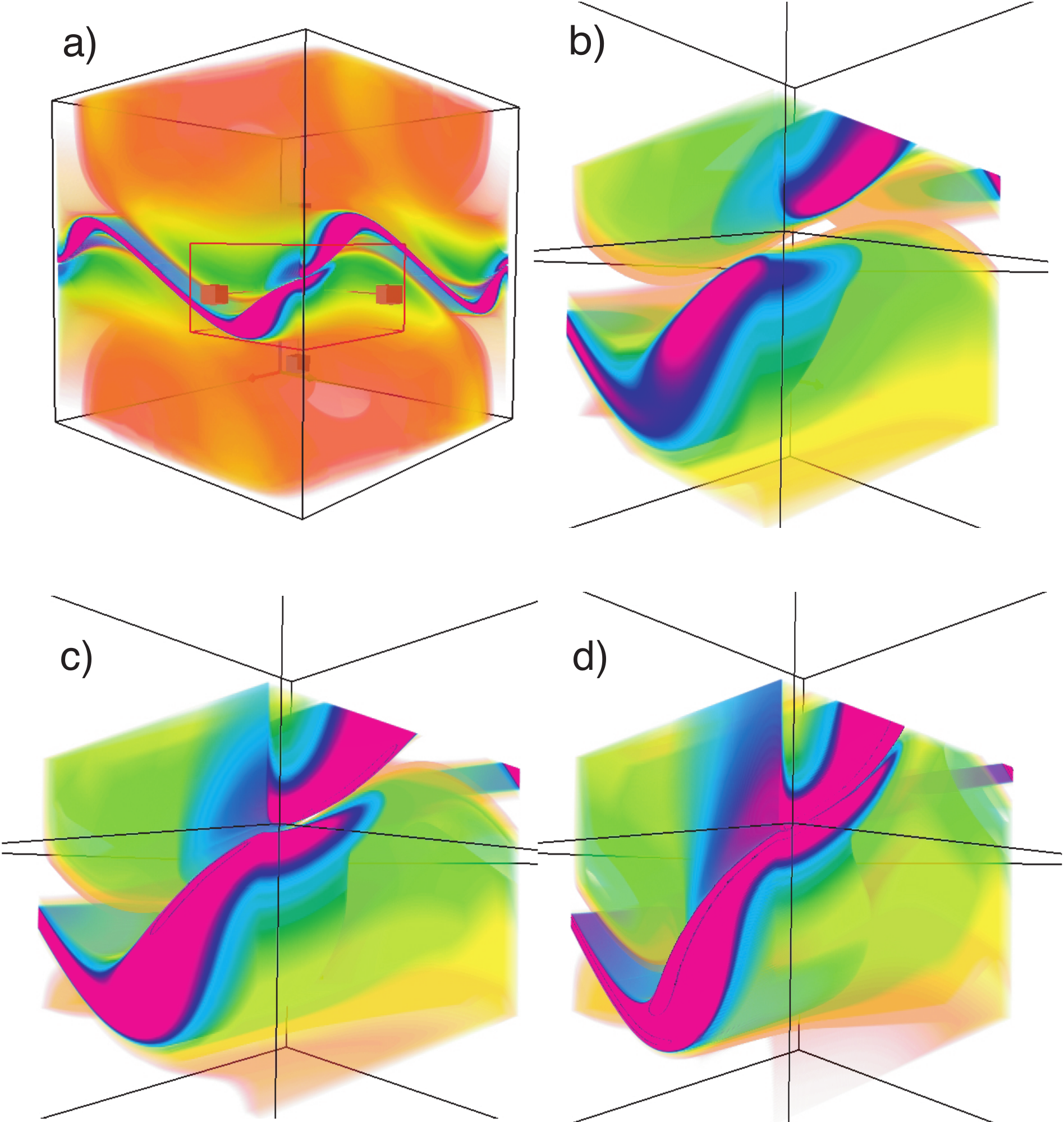}
\caption{(Color online) $3D$ visualization of TG vorticity $|\nabla \times \mbf{u}|$ at resolution $4096^3$: a) full impermeable box $0\leq x \leq \pi$, $0\leq y\leq \pi$ and $0\leq z\leq \pi$ at $t=3.75$.  Zooms over the subbox marked near $x=y=\pi$, $z=\pi/2$ are displayed in b) at $t=3.5$, in c) at $t=3.75$ and in d) at $t=4.0$.} 
\label{Fig:Vort_3D_Viz}
\end{center}
\end{figure}
It is apparent in Fig. \ref{Fig:Energy_Spectra_maxvort}(a) that resolution-dependent even-odd oscillations are present, at certain times, on the TG energy spectrum. Note that this behavior is produced when the tail of the spectrum rises above the round-off error $\sim 10^{-32}$. This phenomenon can be explained in terms of a {\it resonance} \cite{Samriddhi}, along the lines developed in reference \cite{FrischTygers2011}. In practice we will deal with this problem by averaging the spectrum over shells of width
$\Delta k=2$. Apart from this it can be seen that spectra computed using different resolutions are in good agreement for all times.

In contrast, it is visible in Fig. \ref{Fig:Energy_Spectra_maxvort}(b) that the maximum of vorticity $\left\| \bomega(\cdot,t) \right\|_{\infty}$ computed at different resolutions are in agreement only up to some resolution-dependent time (see the inset). The fact that $\left\| \bomega(\cdot,t) \right\|_{\infty}$ at a given time $t>3.7$ decreases if one truncates the higher wavenumbers of the velocity field (see Fig. \ref{Fig:Energy_Spectra_maxvort}(b)) strongly suggests that $\left\| \bomega(\cdot,t) \right\|_{\infty}$ has significant contributions coming from high-wavenumbers modes. This forms the basis of the heuristic argument presented below in Sec. \ref{subsec:Heur}.

Figure \ref{Fig:Vort_3D_Viz} shows $3D$ visualizations (using the VAPOR \footnote{http://www.vapor.ucar.edu} software) of the high vorticity regions in the impermeable box, corresponding to the $4096^3$ run at late times. A thin vortex sheet is apparent in Fig.\ref{Fig:Vort_3D_Viz}(a) on the vertical faces $x=0$, $\pi$ and $y=0$, $\pi$ of the impermeable box.

The emergence of this thin vortex sheet is well understood by simple dynamical arguments about the flow on 
the faces of the impermeable box that were first given in reference \cite{BRACHET:1983p4817}. 
We now briefly review these arguments. 
The initial vortex on the bottom face 
is first forced by centrifugal action to spiral outwards toward the edges and then up the side faces.
A corresponding outflow on 
the top face and downflow from the top edges onto the side faces leads to a 
convergence of fluid near the horizontal centreline of each side face, from where it 
is forced back into the centre of the box and subsequently back to the top and bottom faces. 
The vorticity on the side faces is efficiently produced in the zone of convergence, 
and builds up rapidly into a vortex sheet (see Figs. 1 and 2 of reference \cite{BRACHET:1983p4817} and Fig. 8 of reference \cite{Brachet1992}).

While these considerations explain the presence of the thin vortex sheet in Fig.\ref{Fig:Vort_3D_Viz}(a), the  dynamics presented in Fig.\ref{Fig:Vort_3D_Viz}(b-d) also involves the collision of vortex sheets happening near the edge $x=y=\pi$, close to $z=\pi/2$. Note that, as stated above in Sec. \ref{subsec:symm}, the vortex lines are perpendicular to the faces of the impermeable box. Thus, because the collision takes place near an edge, the corresponding vortex lines must be highly curved, with strong variations of the direction of vorticity.
The geometric constraints on potential singularities posed by the evolution of the direction of vorticity developed in references \cite{ConstFetMaj96,Denghouyu2005,Denghouyu2006} could be applied to the situation described in Fig. \ref{Fig:Vort_3D_Viz}. However, such an analysis goes beyond the BKM theorem and involves extensive post-processing of very large datasets. This task is thus left for further work and we concentrate here on simple BKM diagnostics for the vorticity supremum and analyticity strip analysis of energy spectra.

\subsection{Analyticity-strip analysis of energy spectra}\label{Sec:NumEFits}

The analyticity-strip method \cite{SulemSulemFrisch1983} is based on the fact that when the velocity field is analytic in space the energy  spectrum satisfies $E(k,t)  \propto e^{- 2\, k\,\delta(t)}$ in the asymptotic `ultraviolet region' $k \gg 1,$
with a proportionality factor that may contain an algebraic decay in $k,$ a multiplicative function of time and, depending on the complexity of the physical flow, even an oscillatory (in $k$) modulation \cite{CichowlasBrachet2005}.
\begin{figure}[htbp]
\begin{center}
\includegraphics[height=7.cm]{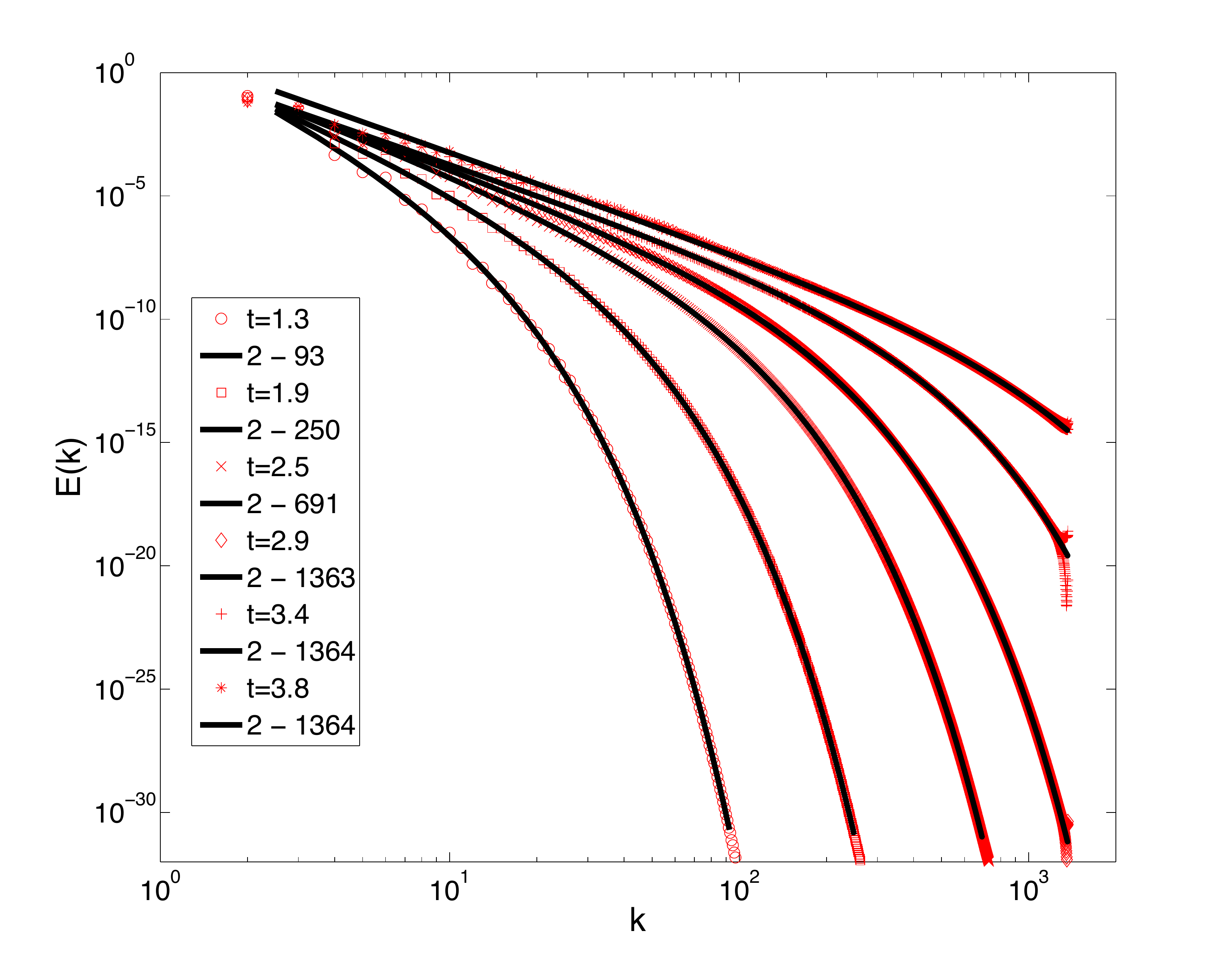}
\caption{(Color online) Comparison of fit  \eqref{eq_fitTG} (solid black line) and spectrum at resolution $4096^3$ (red markers); times and fit intervals are indicated in the legend.}
\label{Fig:Fit_comp}
\end{center}
\end{figure}

\begin{figure*}[htbp]
\begin{center}
\includegraphics[height=12.cm]{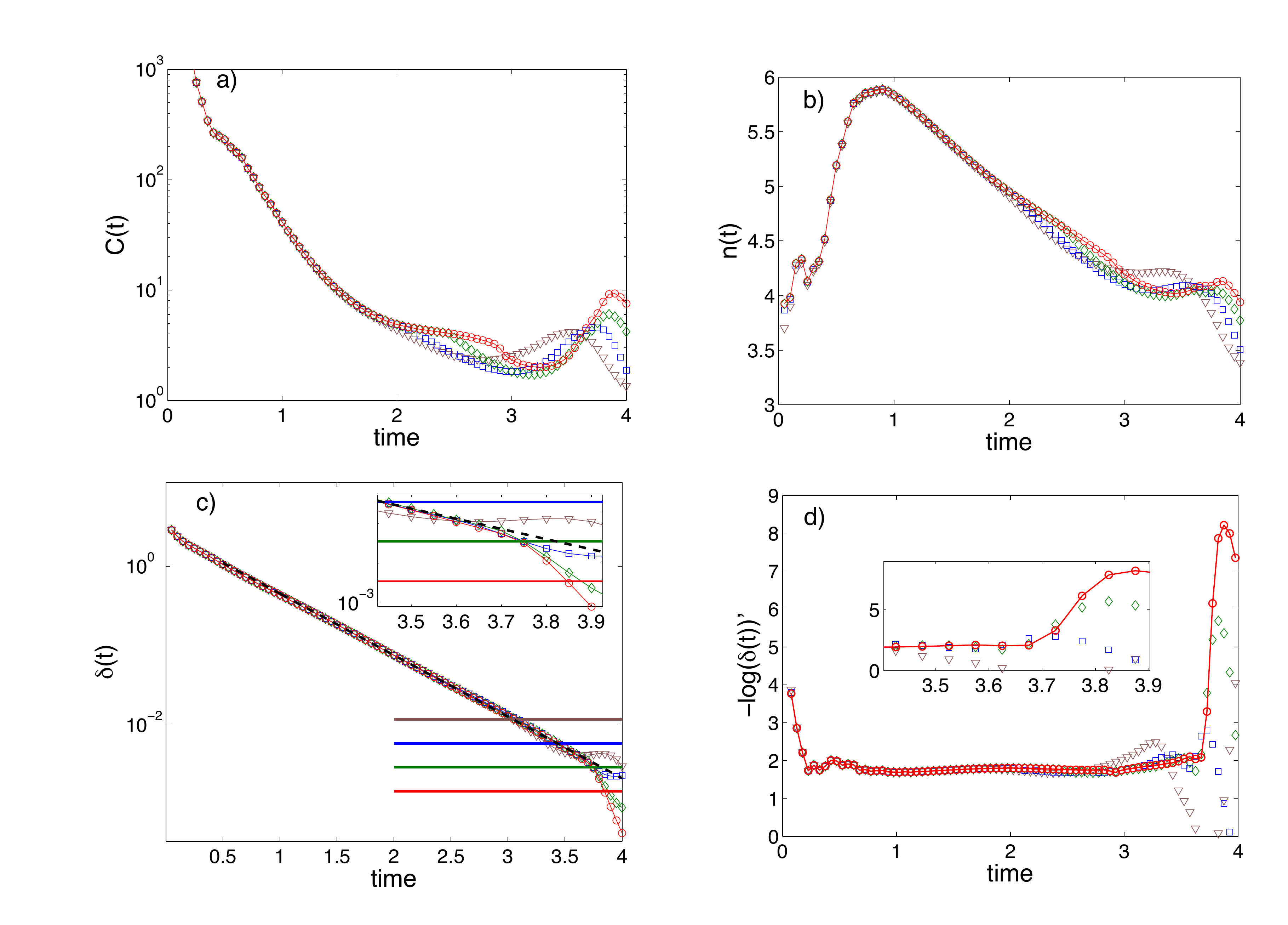} 
\caption{(Color online) Time evolution of energy spectrum fit parameters (see Eq. \eqref{eq_fitTG} and Fig. \ref{Fig:Fit_comp}(a)) constant $C$; b) prefactor $n$;  c) decrement $\delta$ (horizontal lines : $\delta k_{\max}=2$, dashed black line: exponential law \eqref{eq:expdelta});  d) decay rate $-d(\ln(\delta(t))/dt$. Results corresponding to different resolutions are displayed together: $512^3$ (brown triangles), $1024^3$ (blue squares), $2048^3$ (green diamonds) and $4096^3$ (red circles).\label{Fig:Fit_Evolution}}
\end{center}
\end{figure*}
The basic idea is thus to assume that $E(k,t)$ can be well approximated by a function of the form
\begin{equation*}
E(k,t) \approx C(t) \,k^{-n(t)}\,e^{- 2\, k\,\delta(t)}\,,
\end{equation*}
in some wave numbers interval between $1$ and $k_{\max} = \lfloor N/3\rfloor$ (the maximum wavenumber permitted by the numerical resolution $N$).
The common procedure to determine $C(t), n(t), \delta(t)$ is to perform a least-square fit at each time $t$ on
the logarithm of the energy spectrum $E(k,t)$, using the functional
form
\begin{equation}
\label{eq_fitTG}
\ln E(k,t) = \ln C(t) - n(t)\, \ln k  - 2 k \,\delta(t) ~.
\end{equation}
The error on the fit interval ${k_1\leq k\leq k_2}$,
\begin{equation*}
\chi ^2(t)={\displaystyle \sum\limits_{k=k_1}^{k_2}}\left(\ln E(k,t) - \ln C(t) + n(t)\, \ln k  + 2 k \,\delta(t) \right)^2,
\end{equation*}
is minimized by solving the equations
$\partial \chi ^2/\partial C =0$, $\partial \chi ^2/\partial n =0$ and $\partial \chi ^2/\partial \delta =0$.
Note that these equations are linear in the parameters $\ln C(t)$, $n(t)$ and $\delta(t).$

The transient oscillations of the
energy spectrum observed at the highest wavenumbers (see above Fig. \ref{Fig:Energy_Spectra_maxvort}(a) are eliminated by averaging the TG spectrum on shells of width
$\Delta k=2$ before performing the fit \citep{BRACHET:1983p4817}.

We present in Fig. \ref{Fig:Fit_comp}, examples of TG energy spectra fitted in such a way on the intervals $2<k<\min(k^*,k_{\max})$, where $k^*= \inf _{E(k)< 10^{-32}} (k)$ denotes the beginning of round off noise. It is apparent that the fits are globally of a good quality.

The time evolution of the fit parameters  $C$, $\delta$ and $n$ computed at different resolutions are displayed in Fig. \ref{Fig:Fit_Evolution}.
The measure of the fit parameters  is reliable as long as $\delta(t)$ remains larger than a few mesh sizes, a condition required for the smallest scales to be accurately resolved and spectral convergence ensured. Thus the dimensionless quantity $\delta k_{\max}$ is a measure of spectral convergence.

It is conventional \cite{BRACHET:1983p4817} to define a `reliability time' $T_{\mathrm{rel}}$ by the condition
\begin{equation}
\label{eq:trel}
\delta(T_{\mathrm{rel}}) k_{\max}=2 \, ,
\end{equation}
and to say that the numerical simulation is reliable for times $t\leq T_{\mathrm{rel}}$.
This reliability time can be extended only by increasing the spatial resolution available for the simulation, so the more computer power is available the larger is the reliability time.

The resolution-dependent reliability condition \eqref{eq:trel} is marked by the horizontal lines in Fig. \ref{Fig:Fit_Evolution}(c).
The exponential law
\begin{equation}
\label{eq:expdelta}
\delta(t)\sim 2.70 \, e^{-t/0.56}  ~,
\end{equation}
that was previously reported at resolution $2048^3$ in reference \cite{CichowlasBrachet2005}
is also indicated in Fig. \ref{Fig:Fit_Evolution}(c) by a dashed black line.
It is thus apparent that our lower-resolution results well reproduce the previous computations that were discussed above in Sec. \ref{Sec:Intro} (see text preceding references \cite{BRACHET:1983p4817,Brachet1992,CichowlasBrachet2005}).

In Table \ref{Tab:table_rel}, the reliability time \eqref{eq:trel} obtained from the fit parameter $\delta$ of Fig. \ref{Fig:Fit_Evolution} is compared with the reliability time stemming from the exponential behavior \eqref{eq:expdelta}.
\begin{table}[h]
\begin{tabular}{|c||c|c|}
	\hline
Resolution &  $T_{\mathrm{rel}}$ (exponential law)    &   $T_{\mathrm{rel}}$ (fit)  \\
	\hline
$512^3$    & 3.05 & 3.05 \\
$1024^3$   & 3.43 & 3.44 \\
$2048^3$   & 3.82 & 3.75 \\
$4096^3$   & 4.21 & 3.85 \\
	\hline
\end{tabular}
\caption{Reliability time \eqref{eq:trel} deduced from the exponential behavior \eqref{eq:expdelta} compared with the reliability time obtained from the fit parameter $\delta$ of Fig. \ref{Fig:Fit_Evolution}.}
\label{Tab:table_rel}
\end{table}
It is apparent by inspection of the Table that the reliability time of
our new $4096^3$ results is markedly smaller than that deduced from the exponential law \eqref{eq:expdelta}; the latter wrongly predicts that simulations at this resolution should be reliable until $t=4.21$. The departure from the exponential behavior is also visible on the inset in Fig. \ref{Fig:Fit_Evolution}(c).

In order to capture this change of behavior more quantitatively the logarithmic decay rate $-d\ln(\delta)/dt$, computed using finite differences in time, is displayed in Fig. \ref{Fig:Fit_Evolution}(d).
A clear change in trend is apparent around $t=3.7.$ where the logarithmic decay rate abruptly changes from a value near $2$ to a value near $8$.
Note that this change of behavior happens at a time that is reliable at resolution $4096^3$ (see insets in Fig. \ref{Fig:Fit_Evolution}(c) and \ref{Fig:Fit_Evolution}(d). Interestingly, this time is close to the reliability time of the $2048^3$ simulation. Therefore, the new behavior of accelerated decay for times $t>3.7$ can only be suggested by the $2048^3$ data and is here demonstrated for the first time by our $4096^3$-resolution data.
This acceleration of the decay rate of $\delta(t)$ is important because
if \eqref{eq:expdelta} could be safely extrapolated to later times then the Taylor-Green vortex would never develop a real singularity \cite{frischbook}.

Let us conclude this section by showing that the new behavior does not depend on the wavenumber interval chosen to perform the fits.

Indeed, by close inspection of the top curve in Fig. \ref{Fig:Fit_comp} one can see that a small amount of systematic errors are present at the lowest ($k<100$) wavenumbers for large times. Excluding the lowest wavenumbers from the fits results in less errors (data not shown).
In Table \ref{Tab:table_int}, the result of fits performed on the subinterval $103< k< k_{\max}$ are compared with those on the full interval $3< k< k_{\max}$ that was used up to now .
\begin{table}[h]
\begin{tabular}{|c||c|c|c|c|}
	\hline
Time &  $n$   &   $n$   &  $10^{3} \times\delta$     &   $10^{3} \times \delta$ \\
         &$3 - k_{\max}$&  $103 - k_{\max}$   &$3  -  k_{\max}$&  $103 -  k_{\max}$\\
	\hline
3.6  & 4.07  & 3.95  & 4.13  &  4.22 \\
3.65 & 4.09  & 4.05  & 3.73  &  3.75 \\
3.7  & 4.09  & 4.14  & 3.36  &  3.31 \\
3.75 & 4.09  & 4.19  & 2.85  &  2.76 \\
3.8  & 4.12  & 4.29  & 2.10  &  1.95 \\
3.85 & 4.13  & 4.34  & 1.41  &  1.22 \\
3.9  & 4.09  & 4.34  & 0.94  &  0.71 \\
	\hline
\end{tabular}
\caption{Time evolution of fit parameters $n$ and $\delta$ (see Eq. \eqref{eq_fitTG}) on full interval $3< k< k_{\max}$ (same as in Fig. \ref{Fig:Fit_Evolution}) compared with fits on subinterval $103< k< k_{\max}$.}
\label{Tab:table_int}
\end{table}
It can be checked on the table that the departure from the exponential law is not dependent on the interval chosen to perform the fit.
The values of $n$ are also in agreement with previously published data \cite{CichowlasBrachet2005}.

\subsection{BKM analysis of vorticity maximum} \label{subsec:fit_methods_omegas}
In this section we look for eventual singular behavior by focusing on the time-dependence of the
TG data for the vorticity supremum $||\bomega||_\infty(t)$ that is displayed above in Fig. \ref{Fig:Energy_Spectra_maxvort}(b).
The BKM theorem \cite{BKM84} states that blowup (if it takes place) requires the time-integral of the supremum of the vorticity to become infinite.
Our analysis method, first introduced in \cite{BustamanteKerr2008}, amounts to look at evidence of power-law behavior in the numerical time series for $||\bomega||_\infty(t)$ and see if the computed exponent is compatible with blowup of the time integral of $||\bomega||_\infty(t)$.
We now proceed to briefly recall the method.

Let $f(t)$ be the quantity to be studied. In order to test if it might blow up or go to zero in a finite time, we produce, locally in time, fits of power law behavior of the form
\begin{equation}
\label{eq:sing}
f(t) \approx c {{\left(T_* -t \right) }^{\gamma}}\,,
\end{equation}
and we study the `instantaneous' or running estimates for $\gamma$ and $T_*$ as a function of time.

The local fits are done as follows: we first produce the new function
\begin{equation}
\label{eq:localfit}
g(t)=\left(\frac{d \ln f(t)}{dt}\right)^{-1}=f(t)/f'(t).
\end{equation}
If $f(t)$ is of the form (\ref{eq:sing}) then our new function satisfies $g(t) \approx (T_*-t)/\gamma.$ Therefore, a linear fit of $g(t)$ will give $T_*$ and $\gamma$. More explicitly, we have the local expressions
\begin{equation}
\label{eq:gamma}
\gamma(t)=\left({1-\frac{f(t)\,f''(t)}{{f'(t)}^2}} \right)^{-1},
\end{equation}
and
\begin{equation}
\label{eq:Tc}
T_*(t)=t + \frac{f(t)\,f'(t)}{ f(t)\,f''(t)-{f'(t)}^2}.
\end{equation}
The latter local expressions can be used with any suitable fit method of the data, not necessarily linear fits.

In practice, as our time series are given on an equally spaced temporal grid, we proceed in the following straightforward manner. First we compute $\ln(f(t))$, then we use centered finite differences to estimate its derivative. Inverting this data furnishes estimates of $g(t)$ at the midpoints. Using again centered finite differences produces estimates of $1/\gamma$ on the original grid, thus allowing the determination of local estimates for both $T_*$ and $\gamma$. Note that this algorithm basically amounts to a local $3$-point nonlinear fit.

\begin{figure}[htbp]
\begin{center}
\includegraphics[height=6.5cm]{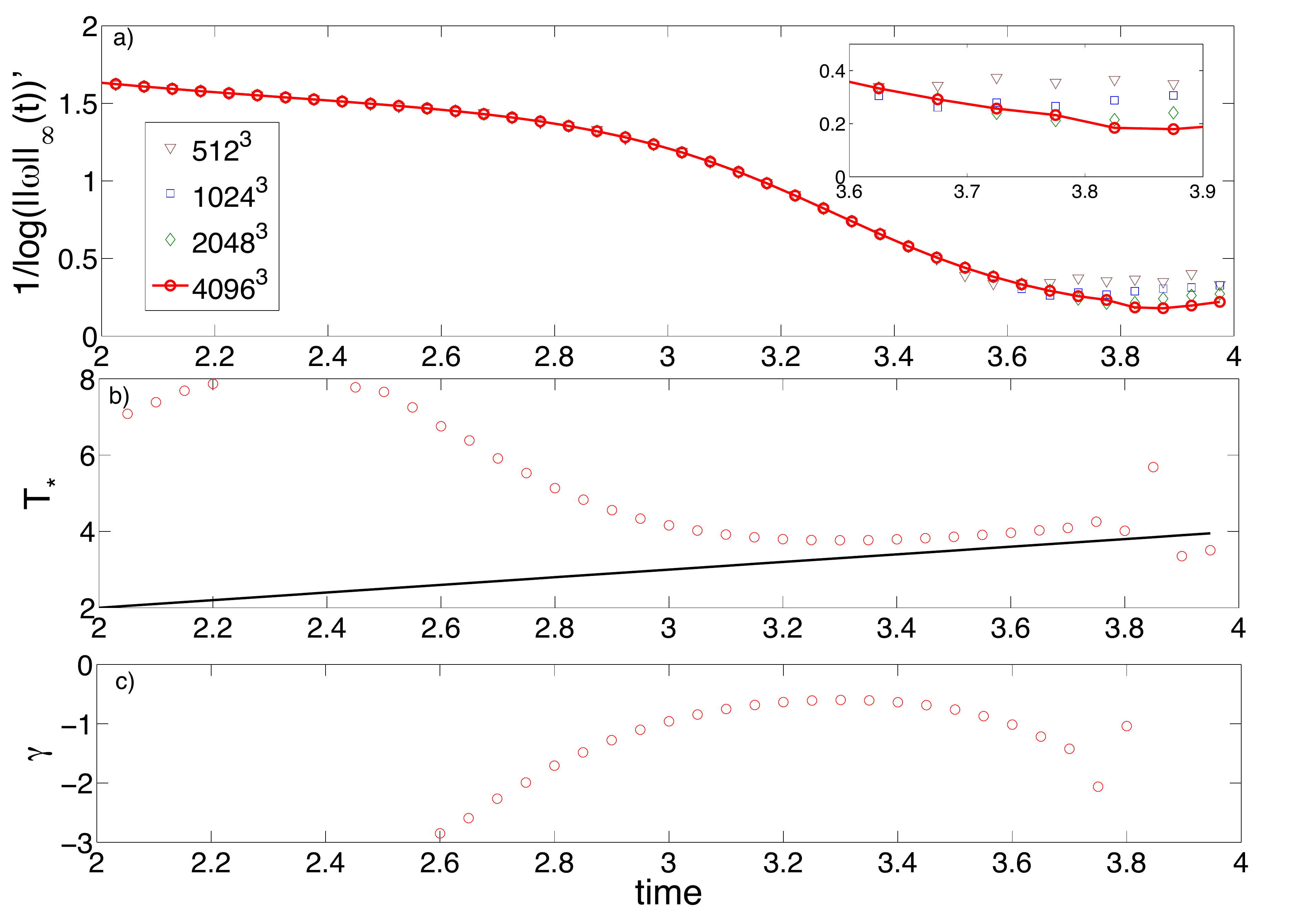}
\caption{(Color online) Time evolution of a) inverse logarithmic derivative \eqref{eq:localfit} at all resolutions (see legend); b) extrapolated $T_*$ \eqref{eq:Tc} (solid black line: $T_*=t$) and c) running value of $\gamma$ \eqref{eq:gamma}, both only at resolution $4096^3$ (red circles).}
\label{Fig:VortSup}
\end{center}
\end{figure}

The values of $g(t)$, $T_*(t)$ and $\gamma(t)$ obtained in this way from the TG data for the vorticity supremum $||\omega ||_\infty$ are displayed in Fig. \ref{Fig:VortSup}.
It is apparent that $g(t)$ presents an inflection point around $t=3.3$ corresponding to a maximum value of $\gamma$ that is above $-1$. Thus local in time power-law extrapolations around $t=3.3$ are inconsistent with the BKM theorem that requires $\gamma\leq -1$. However, when $t$ is larger than $3.6$, the value of $\gamma$ goes below $-1$ and thus becomes compatible with BKM.

On the other hand, there is no sign that the data values of  $\gamma$ and $T_*$ are settling down into constants, corresponding to a simple power-law behavior.

Recall (see Section \ref{Sec:NumEFits}) that the last reliable value of $||\omega ||_\infty$ at resolution $4096^3$ is at $t=3.85$.
Thus, due to our $3$-point extrapolation method, the last reliable data point is at $t=3.825$ in Fig. \ref{Fig:VortSup}(a) and  at $t=3.8$ in Figs. \ref{Fig:VortSup}(b) and \ref{Fig:VortSup}(c). The data corresponding to $\gamma$ and $T_*$ are also displayed in Table \ref{Tab:table_omegasup}.
\begin{table}[h]
\begin{tabular}{|c||c|c|}
	\hline
Time &  $\gamma$   &     $T_*$     \\
	\hline
3.7   & -1.42  & 4.09  \\
3.75  & -2.06   & 4.26  \\
3.8   & -1.04   & 4.02 \\
	\hline
\end{tabular}
\caption{Power-law fit parameters $\gamma$ and $T_*$ (see Eq. \eqref{eq:sing}) for the vorticity supremum $||\omega ||_\infty$ determined at resolution $4096^3$ (see Fig. \ref{Fig:VortSup}(b) and Fig. \ref{Fig:VortSup}(c)).}
\label{Tab:table_omegasup}
\end{table}

Thus our conclusion for this section is that although our late-time reliable data for $||\omega ||_\infty(t)$ shows  $\gamma(t)<-1$ and is therefore not inconsistent with BKM, clear power-law behavior of $||\omega ||_\infty(t)$ is not achieved.

\section{Bridging analyticity-strip method and BKM Theorem}
\label{Sec:As_Bkm}

\subsection{Motivation and simple estimates}\label{subsec:Heur}
The vorticity maximum $\left\| \bomega(\cdot,t) \right\|_{\infty}$ was found to decrease when the resolution is reduced at any given time $t>3.7$ (see the above discussion following Fig. \ref{Fig:Energy_Spectra_maxvort}(b)). This strongly suggests that, in this late-time regime, $\left\| \bomega(\cdot,t) \right\|_{\infty}$ has significant contributions coming from high-wavenumbers modes. In this context, the following short heuristic argument is provided as a motivation for the more rigorous mathematical results to follow.

Consider the well-known Sobolev inequality, which can be derived using the same hypotheses as in Lemma \ref{lem:main} below:
\begin{eqnarray}
\label{eq:ineq_old}
 \left\| \bomega(\cdot,t) \right\|_{\infty} &\leq& C_\epsilon \, \sqrt{2\,\Omega_{\epsilon +5/2}(t)}\,,\quad \forall \,t \in [0,T)\,.
\end{eqnarray}
This bound is valid for any $\epsilon>0,$ where
\begin{equation}
C_\epsilon \equiv \sqrt{\displaystyle \sum_{\mbf{k} \in \mathbb{Z}_{\mathrm{odd}}^3 \cup \mathbb{Z}_{\mathrm{even}}^3 \setminus \{\mbf{0}\}} \left|\mbf{k}\right|^{-3-2\,\epsilon}}\  ,\label{eq:Cepsilon}
\end{equation}
and
$\Omega_p$ is defined by
\begin{eqnarray}
 \label{eq:Sobolev}
\Omega_p(t) \equiv   \frac{1}{2}{\displaystyle\sum\limits_{\mbf{k}\in \mathbb{Z}_{\mathrm{odd}}^3 \cup \mathbb{Z}_{\mathrm{even}}^3 }} |\mbf{k}|^{2p}\,|\widehat{\mbf{u}}(\mbf{k},t)|^2 \,.
\end{eqnarray}
Notice that  $2 \Omega_p$ is the square of the Sobolev seminorm $\left | \mbf{u}(\cdot,t) \right |_{H^p}.$

Motivated by the numerical results of Section \ref{Sec:NumEFits}, let us assume, at a given time $t$, a behavior of the energy spectrum \eqref{eq:spectrum} of the type
\begin{equation}
\label{eq:heuristic energy}
E(k)\sim k^{-n}e^{-2 \delta k}.
\end{equation}
Notice that $n$ and $\delta$ are functions of time. When $n<6$ and $\delta$ tends to zero, this gives a UV-divergence:
\begin{equation*}
\Omega_{\epsilon +5/2} \sim \int_1^\infty k^{5+2 \epsilon-n}e^{-2 \delta k} dk \sim \delta^{-6+n-2 \epsilon}\,. 
\end{equation*}
Plugging this into the bound \eqref{eq:ineq_old}, and using the BKM theorem, we get $\int^{T_*} \delta(t)^{-3+\frac{n}{2}- \epsilon} dt =\infty,$ where $T_*$ is the hypothetical singularity time.

At this point, again motivated by our numerical results, we assume $n=\mathrm{const.} < 6$ and assume a power-law behavior for the analyticity-strip width of the form
\begin{equation*}
\delta(t) \propto (T_{*}-t)^{\Gamma}\,, 
\end{equation*}
where $\Gamma>0$ is a constant. Replacing this into the above integral we conclude that
\begin{equation*}
\int^{T_*} (T_*-t)^{(-3+\frac{n}{2}- \epsilon)\Gamma} dt = \infty\,,
\end{equation*}
i.e., a finite-time singularity can be attained only if the exponents satisfy $(-3+\frac{n}{2}- \epsilon)\Gamma \leq -1$ for any $\epsilon > 0\,.$ Taking the limit $\epsilon \to 0$ we deduce finally
\begin{equation*}
\Gamma \geq \frac{2}{6-n}\,. 
\end{equation*}
In words: \textbf{``if the analyticity-strip width $\delta(t)$ goes to zero as a power law, then the exponent must be greater than or equal to $\frac{2}{6-n}\,.$''}

The main difficulty to overcome in order to materialize the above heuristic arguments into a firm basis, is that the common Sobolev bound \eqref{eq:ineq_old} has a problem at $\epsilon = 0:$ the constant $C_\epsilon$ is equal to infinity there, so taking the limit as we did is not fully justified. We provide the solution to this problem by finding a new rigorous bound, sharper than the common Sobolev bound, and which gives the same optimal exponents without a divergent constant.

The second difficulty is that the assumed behavior for the energy spectrum \eqref{eq:heuristic energy}, commonly used in the analyticity-strip method, is a very strong condition and does not hold uniformly for $k \in \mathbb{N}.$ In fact, the evidence in analytically solvable models such as the 1D Burgers equation, is that the behavior \eqref{eq:heuristic energy} holds with some exponents $n$ and $\delta$ in the region $k \gg \delta^{-1},$ (large-$k$ asymptotic limit), and the behavior $E(k,t) \sim k^{-\tilde{n}}$ holds in the region $1 \leq k \ll \delta^{-1},$ with $\tilde{n} < n.$ We provide the solution to this lack of uniformity by introducing a ``working hypothesis'' which is a uniform-in-$k$ inequality for the energy spectrum, that still retains the spirit of the analyticity-strip method. The working hypothesis is verified for the case of 1D Burgers equation (see below the discussion at the end of Sec. \ref{Sec:Conclusion}).

\subsection{Mathematical preliminaries}

\subsubsection{BKM theorem}
We assume the usual hypotheses of the Beale-Kato-Majda (BKM) theorem. Let $T$ denote, from here on, a generic time so that the velocity field $\mbf{u} \in C([0,T);H^p)\cap
C^1([0,T);H^{p-1})\,,\,\,p\geq3,$  so in particular the quantities defined in \eqref{eq:Sobolev} are bounded for $p\geq3:$
\begin{eqnarray*}
\Omega_p(t) \leq c_p\,,\,\,\forall \,t\in[0,T)\,.
\end{eqnarray*}
The BKM theorem \cite{BKM84} states that the
assumed regularity of the velocity field can be extended up to and including the time $T$ \mbox{if and only if}
\mbox{$\tau(T) \equiv \int_{0}^T ||\bomega(\cdot,t)||_{\infty} ~dt < \infty.$} By `regular up to and including the time
$T$' we mean $\mbf{u} \in C([0,T];H^p)\cap C^1([0,T];H^{p-1}),\,p\geq3$.

\begin{defn}
\label{defn:T*}
We define the maximal time of regularity $T_* \in (0,\infty]$ as the earliest time for which $\mbf{u}$ ceases to be in  $C([0,T];H^p)\cap
C^1([0,T];H^{p-1})\,,\,\,p\geq3.$

\noindent If $T_* < \infty$ we speak of a finite-time singularity.
\end{defn}

With this definition, we conclude that the time integral appearing in the BKM theorem converges for all $T<T_*$ and diverges at $T=T_*$: $\int_{0}^{T_*} ||\bomega(\cdot,t)||_{\infty} ~dt = \infty\,.$

\subsubsection{Working hypothesis for energy spectrum}
An implicit assumption of the analyticity-strip method is the existence of the Fourier components of the solution of the 3D Euler equations. Taylor-Green (TG) symmetries imply that only modes with even-even-even and odd-odd-odd wavenumber components are present (see Section \ref{subsec:symm}). The appropriate definition of the energy spectrum is thus

\begin{defn}
\label{defn:spectrum}
The kinetic energy spectrum $E(k,t)$ is defined as the sum of squares of modulus of Fourier coefficients over spherical shells
\begin{equation}
\label{eq:spectrumTG}
E(k,t) = \frac{1}{2} {\displaystyle \sum_{\mbf{k} \in  \mathbb{Z}_{\mathrm{odd}}^3 \cup \mathbb{Z}_{\mathrm{even}}^3 \atop k-1/2  < |\mbf{k}| < k+1/2 }} |\widehat{\mbf{u}}({\bf k},t)|^2.
\end{equation}
\end{defn}

It is easy to check that the TG symmetries imply that $E(0,t)=E(1,t)=0 \quad \forall t \in [0,T_*).$ Numerical observations (see reference \cite{CichowlasBrachet2005} and Section \ref{Sec:Numerics_Classical} above) lead us to formulate the following working hypothesis that will be used to bound the energy spectra:\\

\begin{hypo}[Working hypothesis]
\label{hypo:working}
From here on, we will assume that
there exist a constant $M>0$ and positive functions $n_0(t), \delta_0(t),$ continuous on $[0,T_*),$ such that for all times $t \in [0,T_*)$ and all $k \in \mathbb{Z}, k \geq 2$ we have
\begin{equation}
\label{eq:fit_bound}
E(k,t) \leq M \,k^{-n_0(t)}\,e^{- 2\, k\,\delta_0(t)}\,.
\end{equation}
\end{hypo}

\noindent \textbf{Remarks.} (i) The working hypothesis is consistent with the hypotheses of the BKM theorem. \mbox{(ii) The} working hypothesis is an inequality defined globally in $k$ and \emph{is not a large-$k$ asymptotic expansion}. Furthermore, a large-$k$ asymptotic expansion is typically of the form $E(k,t) = C_1(t) k^{-n_1(t)}\,e^{- 2\, k\,\delta_1(t)}$ and has, in contrast to \eqref{eq:fit_bound}, a time-dependent constant $C_1(t)$. Nevertheless, asymptotic results can be used to establish the working hypothesis in special cases such as the $1D$-inviscid Burgers equation (see the discussion below, at the end of Sec. \ref{Sec:Conclusion}). \mbox{(iii) The} numerically-obtained fits of the analyticity-strip method $E(k,t) \approx C(t) k^{-n(t)}\,e^{-2 k \delta(t)}$ are similarly related to the working hypothesis. Notice that these fits are obtained over a finite range of values of wavenumber $k$, so they give only partial information. At early times, when the analyticity-strip width $\delta$ is big so that $\delta k \gg 1,$ one is in the ``large-$k$ asymptotic limit''. At late times, when $\delta$ becomes of the order of the highest resolved wavenumber $k_{\max},$ we have $\delta k \lessapprox 1$ and thus the fits represent the ``small-$k$ range''. The relations $n(t) \geq n_0(t)$ and $\delta(t) \geq \delta_0(t)$ are required for consistency with the working hypothesis. In practice, we will use the numerically obtained $n(t)$ and $\delta(t)$ to estimate $n_0(t)$ and $\delta_0(t)$.

\subsubsection{Classification of solutions in terms of regularity}

We see from Definition \ref{defn:T*} that a finite-time singularity is defined by the condition $T_* < \infty.$ Combining this with the working hypothesis, a finite-time singularity can occur only if $\lim_{t \to T_*} \delta_0(t) = 0.$ Amongst all possible continuous positive functions $\delta_0(t)$ that tend to zero as $t \to T_*$ we will consider, to simplify the analysis, only power-law type of functions.

\begin{defn}
\label{defn:power-law sing}
A solution of the 3D Euler equations satisfying the working hypothesis \eqref{eq:fit_bound}, is said to have a finite-time singularity of power-law type, with power $\Gamma > 0,$ iff the working hypothesis admits a function $\delta_0(t)$ that behaves, near $t = T_*,$ as
\begin{equation*}
 \delta_0(t)  \propto (T_* - t)^\Gamma.
\end{equation*}

\end{defn}

We saw in the heuristics Section \ref{subsec:Heur} that if the energy spectrum is of the form $E(k,t) \approx C(t) k^{-n(t)}\,e^{-2 k \delta(t)},$ then the exponent $n(t)$ must be less than 6 in order for a finite-time singularity to occur. This result will be fully formalized in Section \ref{subsec:main_results}, but first we need to define two types of solutions in terms of the behavior of the exponent $n_0(t)$  appearing in the working hypothesis.

\begin{defn}
\label{defn:strong_vs_mild}
A solution of the 3D Euler equations satisfying the working hypothesis \eqref{eq:fit_bound}, is said to be of strong regularity iff the working hypothesis admits an exponent $n_0(t)$ such that $\displaystyle \liminf_{t \to T_*} n_0(t) > 6.$ Otherwise, i.e. if all the exponents admitted by the working hypothesis satisfy $\displaystyle \liminf_{t \to T_*} n_0(t) \leq 6,$ the solution is said to be of mild regularity.
\end{defn}

The reason for the name ``strong'' is due to the following lemma (to be proved in Section \ref{subsec:main_results}):

\begin{lem}
\label{lem:strong}
Let a solution of the 3D Euler equations satisfying the working hypothesis \eqref{eq:fit_bound}, be of strong regularity. Then the solution has no finite-time singularity.
\end{lem}

This lemma's assertion is basically the same as the well-known fact that there cannot be a finite-time loss of analytic regularity without loss of $C^1$ regularity \cite{BardosBenachour77,KulkaviaVicol2011}.

This result can be used as a validation test for numerical simulations of 3D Euler fluids. If the supremum norm of the vorticity is to grow in time without bound, then the exponent $n_0(t)$ must be well below the critical value $6.$ Fortunately, all reliable numerical simulations that we know of pass this elementary test.

\subsection{Main results linking Beale-Kato-Majda theorem and analyticity-strip method}
\label{subsec:main_results}

\subsubsection{Sharp bound for vorticity}

\begin{lem}
\label{lem:main}
Let $\mbf{u}(\mbf{x},t)$ be a velocity field satisfying the Taylor-Green symmetries and with energy spectrum defined by equation \eref{eq:spectrumTG}. Let $\bomega = \nabla \times \mbf{u}$ be its vorticity, defined on the periodicity domain $D=[0,2\,\pi]^3.$ Then the following inequality is verified for all times $t \in [0,T)$:
\begin{eqnarray}
\label{eq:ineq_1}
 \left\| \bomega(\cdot,t) \right\|_{\infty} &\leq& {\displaystyle \sum_{k=2}^\infty \,\sqrt{2\,k(k+1)\, E(k,t)\,S_k}} \,,
\end{eqnarray}
where $S_k \equiv \#\{\mbf{k} \in  \mathbb{Z}_{\mathrm{odd}}^3 \cup \mathbb{Z}_{\mathrm{even}}^3 : k-1/2 < |\mbf{k}| < k+1/2\}$ is the combined number of lattice points (of the form odd-odd-odd or even-even-even) in a spherical shell of width 1 and radius $k \in \mathbb{Z}_+$.
\end{lem}

\begin{proof}[Proof] The vorticity field is defined in terms of its Fourier components by $\bomega(\mbf{x},t) = {\displaystyle \sum_{\mbf{k} \in \mathbb{Z}_{\mathrm{odd}}^3 \cup \mathbb{Z}_{\mathrm{even}}^3} e^{i \mbf{k}\cdot\mbf{x}} \widehat{\bomega}(\mbf{k},t)}.$ Therefore
\begin{equation}
 \label{eq:bound_vort_1}
\left|\bomega(\mbf{x},t)\right| \leq {\displaystyle \sum_{\mbf{k} \in \mathbb{Z}_{\mathrm{odd}}^3 \cup \mathbb{Z}_{\mathrm{even}}^3} \left|\widehat{\bomega}(\mbf{k},t)\right|},
\end{equation}
for all $\mbf{x} \in D.$ The LHS of this equation can be replaced by the supremum norm. Also, we use the identity $ \left|\widehat{\bomega}(\mbf{k},t)\right| = |\mbf{k}|  \left|\widehat{\mbf{u}}(\mbf{k},t)\right|$ on the RHS and obtain
\begin{equation*}
 \left\| \bomega(\cdot,t) \right\|_{\infty} \leq  {\displaystyle \sum_{\mbf{k} \in \mathbb{Z}_{\mathrm{odd}}^3 \cup \mathbb{Z}_{\mathrm{even}}^3} |\mbf{k}|  \left|\widehat{\mbf{u}}(\mbf{k},t)\right|}. 
\end{equation*}
Assuming that $\mbf{u}$ is regular so the above sum over the lattice converges, we can rewrite the sum over spherical shells of width 1 and radius $k \in \mathbb{Z}_+.$ We get
\begin{equation*}
 \left\| \bomega(\cdot,t) \right\|_{\infty} \leq  {\displaystyle \sum_{k=2}^\infty \left(\sum_{\mbf{k} \in \mathbb{Z}_{\mathrm{odd}}^3 \cup \mathbb{Z}_{\mathrm{even}}^3  \atop k-1/2  < |\mbf{k}| < k+1/2 } |\mbf{k}|  \left|\widehat{\mbf{u}}(\mbf{k},t)\right|\right)}. 
\end{equation*}
We proceed to bound the terms in brackets, for a given $k \in \mathbb{Z}_+.$ First, notice that the highest possible value of $|\mbf{k}|$ is equal to $\sqrt{k(k+1)}.$ We obtain the preliminary result
$ \left\| \bomega(\cdot,t) \right\|_{\infty} \leq  {\displaystyle \sum_{k=2}^\infty \sqrt{k(k+1)} \left(\sum_{\mbf{k} \in \mathbb{Z}_{\mathrm{odd}}^3 \cup \mathbb{Z}_{\mathrm{even}}^3  \atop k-1/2  < |\mbf{k}| < k+1/2 } \left|\widehat{\mbf{u}}(\mbf{k},t)\right|\right)}.$
Second, the remaining sum in brackets is related to the energy spectrum $E(k,t),$ equation (\ref{eq:spectrum}), by virtue of the Cauchy-Schwartz inequality. We have
\begin{equation}
\label{eq:bound_norm}
 \displaystyle \sum_{\mbf{k} \in \mathbb{Z}_{\mathrm{odd}}^3 \cup \mathbb{Z}_{\mathrm{even}}^3  \atop k-1/2  < |\mbf{k}| < k+1/2 } \left|\widehat{\mbf{u}}(\mbf{k},t)\right| \leq \sqrt{2\,E(k,t)} \sqrt{\sum_{\mbf{k} \in \mathbb{Z}_{\mathrm{odd}}^3 \cup \mathbb{Z}_{\mathrm{even}}^3  \atop k-1/2  < |\mbf{k}| < k+1/2 } 1} ,
\end{equation}
which establishes the Lemma. \end{proof}

\noindent \textbf{Remarks.}  The proof is independent of any evolution equation that $\mbf{u}$ might satisfy. Only two inequalities have been used to get the bound (\ref{eq:ineq_1}), and these inequalities are quite sharp:

First, the bound (\ref{eq:bound_vort_1}) is saturated when all phases are equal in the Fourier expansion for the vorticity field at the position of vorticity maximum. This saturation indeed takes place in one-dimensional systems that blow up in a finite time, such as the inviscid Burgers equation (work in progress).

Second, the bound (\ref{eq:bound_norm}) is saturated when all the terms are equal in the sum over the spherical shell of fixed radius $k$. Physically, such saturation should be observed in a fully isotropic scenario, i.e., when the terms $\left|\widehat{\mbf{u}}(\mbf{k},t)\right|^2$ depend more on the wavevector's modulus $|\mbf{k}|$ than on its direction $\mbf{k}/|\mbf{k}|$.

In contrast, the Sobolev bound \eqref{eq:ineq_old} would be saturated only for unphysical scenarios where the energy spectrum $E(k,t)$ has a compact support in $k$-space and is independent of the wavenumber $k$ on that support. Thus the Sobolev bound \eqref{eq:ineq_old} will be less sharp than the new bound \eqref{eq:ineq_1} for any realistic energy spectrum that decays as $k \to \infty.$\\

\noindent \textbf{Practical form.} We provide a more practical form of the sharp bound (\ref{eq:ineq_1}), by noticing that $S_k \approx \pi k^2$ as $k \to \infty$. Under the hypotheses of Lemma \ref{lem:main}, we readily obtain the estimate
\begin{eqnarray}
\label{eq:ineq_2_TYG}
 \left\| \bomega(\cdot,t) \right\|_{\infty} \leq c\,{\displaystyle \sum_{k=2}^\infty k^2\sqrt{E(k,t)}} \,,
\end{eqnarray}
where $c=2\sqrt{11/3}$. This constant was computed by direct inspection of the maximum deviation from the asymptotic formula $S_k \approx \pi k^2$. Although this estimate seems not as sharp as the original one, it will be enough for the practical situation where the analyticity-strip width $\delta(t)$ tends to zero and the main contribution comes from the `ultraviolet region' $k\gg 1.$

\subsubsection{Implications of BKM Theorem: General result}

Let us replace the working hypothesis for the energy spectrum (\ref{eq:fit_bound}) into the bound (\ref{eq:ineq_2_TYG}). The sum over   $k\geq2$ can be written in terms of the so-called polylogarithm function. We obtain the bound
\begin{equation}
\label{eq:ineq_inter}
  \left\| \bomega(\cdot,t) \right\|_{\infty} \leq c\,\sqrt{M}\, \LiTyG{\frac{n_0(t)}{2}-2}{{\mathrm{e}}^{-\delta_0(t)}}\,,
\end{equation}
where $\LiTyG{s}{z}$ is defined by
\begin{equation*}
\LiTyG{s}{z} \equiv \sum_{k=2}^\infty k^{-s} {z^k} = \Li{s}{z} - z\,,
\end{equation*}
 and $\Li{s}{z}$ is the Jonqui\`ere's function (or polylogarithm):
$\Li{s}{z} \equiv \sum_{k=1}^\infty k^{-s} {z^k}\,$.

Combining the bound \eref{eq:ineq_inter} with the BKM theorem we obtain the following

\begin{thm}
 \label{thm:main}
Let a solution of the 3D Euler equations satisfy the Taylor-Green symmetries and the working hypothesis \eref{eq:fit_bound}. Then its maximal regularity time $T_*$ must satisfy
\begin{equation}
 \label{eq:impl_BKM}
\int_0^{T_*} \LiTyG{\frac{n_0(t)}{2}-2} {{\mathrm{e}}^{-\delta_0(t)}}\,dt = \infty.
\end{equation}
\end{thm}

\begin{proof}[Proof]
 The proof is a direct application of the BKM theorem to inequality \eref{eq:ineq_inter}. \end{proof}

At this point it is necessary to state without proof some properties of the polylogarithm:

\begin{lem}
\label{lem:poly}
The polylogarithm function $\Li{p}{z}$ satisfies the following properties:\\

\noindent (i) Let $0 < z < 1$ and let $p, q$ be two non-negative numbers. Then we have $\Li{p}{z} \leq \Li{q}{z} \iff p \geq q\,.$\\

\noindent (ii) Let $|\mu|<2\pi$ and let $r\in \mathbb{R}\setminus \mathbb{Z}_+.$  Then
\begin{equation*}
\Li{r}{e^{\mu}} \approx \Gamma(1-r)\,(-\mu)^{r-1} + {\displaystyle \sum \limits_{k=0}^{\infty}} \frac{\zeta(r-k)}{k!} \mu^k\,, 
\end{equation*}
where $\zeta$ is the Riemann zeta function.\\

\noindent (iii) Let $|\mu|<2\pi$ and let $s \in \mathbb{Z}_+.$ Then
\begin{equation*}
\Li{s}{e^{\mu}} \approx \frac{\mu^{s-1}}{(s-1)!}\left[H_{s-1} - \ln (-\mu)\right] + {\displaystyle \sum \limits_{k=0 \atop k\neq s-1}^{\infty}} \frac{\zeta(s-k)}{k!} \mu^k\,, 
\end{equation*}
where $H_p = {\displaystyle \sum \limits_{h=1}^{p}} \frac{1}{h}$ is the $p$-th harmonic number, with $H_0 = 0.$\\
\end{lem}

We are now ready to prove\\

\noindent \textbf{Lemma \ref{lem:strong}.} \emph{Let a solution of the 3D Euler equations satisfying the working hypothesis \eqref{eq:fit_bound}, be of strong regularity. Then the solution has no finite-time singularity.}

\begin{proof}[Proof] By definition, solutions of strong regularity satisfy the working hypothesis with $\liminf_{t \to T_*} n_0(t) > 6.$ Therefore, using Lemma \ref{lem:poly} (i) on equation \eref{eq:impl_BKM}, we obtain $\int^{T_*} \LiTyG{1+\epsilon} {{\mathrm{e}}^{-\delta_0(t)}}\,dt = \infty\,,$ for some $\epsilon \in (0,1).$ Now, using Lemma \ref{lem:poly} (ii) with $r>1,$ we obtain that the integrand is continuous in time. Therefore $T_*=\infty.$ \end{proof}

\subsubsection{Implications of BKM Theorem: Singularity scenarios}

Theorem \ref{thm:main} represents our `bridge' from analyticity-strip method to BKM Theorem: a singularity of the solution at time $T_*$ can be attained only if the parameters $n_0(t)$ and $\delta_0(t)$ satisfy equation \eref{eq:impl_BKM}.

Recall that for a singularity to occur, the function $\delta_0(t)$ must tend to zero as $t \to T_*.$ The polylogarithm $\LiTyG{\frac{n_0(t)}{2}-2} {{\mathrm{e}}^{-\delta_0(t)}}$ has a branch point at $n_0(t) = 6, \delta_0(t) = 0$ (see Lemma \ref{lem:poly} (iii)) so the asymptotic behavior of the integrand (\ref{eq:impl_BKM}) as $\delta_0(t) \to 0$ depends sensitively on the behavior of the function $n_0(t)$ near the `critical' value $6.$ To avoid this branch point, we introduced solutions with strong and mild regularity in Definition \ref{defn:strong_vs_mild}.

The two following main results exploit the consequences of Theorem \ref{thm:main} in singularity scenarios. They provide us with a criterion on how fast must $\delta_0(t)$ decay to zero in order to sustain a singularity.

\begin{cor}
\label{cor:finite-time-sing}
Let a solution of the 3D Euler equations satisfy the Taylor-Green symmetries and the working hypothesis \eref{eq:fit_bound}. Let the solution be of mild regularity, i.e.,  $\liminf_{t \to T_*} n_0(t) \leq 6$, where $T_*$ is the maximal regularity time. Let $\lim_{t\to T_*}\delta_0(t) = 0.$ Then, $T_*$ satisfies
\begin{equation*}
\int^{T_*} \left(\frac{1}{\delta_0(t)}\right)^{\frac{6-n_-}{2}}\,dt = \infty,
\end{equation*}
for all $n_-$ in $(-\infty, {\displaystyle \liminf_{t \to T_*} n_0(t)}] \cap (-\infty,6).$
\end{cor}
\begin{proof}[Proof]  Let $n_-$ be in $(-\infty, {\displaystyle \liminf_{t \to T_*} n_0(t)}] \cap (-\infty,6).$ From $n_- \leq \liminf_{t \to T_*} n_0(t)$, using Lemma \ref{lem:poly} (i) on equation \eref{eq:impl_BKM} we obtain
\begin{equation*}
\int^{T_*} \LiTyG{\frac{n_-}{2}-2} {{\mathrm{e}}^{-\delta_0(t)}}\,dt = \infty\,. 
\end{equation*}
Now, since $n_-<6$ and the function $\delta_0(t)$ tends to zero as $t \to T_*$, we can use Lemma \ref{lem:poly} (ii) to bound the integrand $\LiTyG{\frac{n_-}{2}-2} {{\mathrm{e}}^{-\delta_0(t)}}$ by a constant times $\left(\frac{1}{\delta_0(t)}\right)^{\frac{6-n_-}{2}},$ which completes the proof. \end{proof}

Finally we consider the hypothetical situation of a finite-time singularity of power-law type, as described in Definition \ref{defn:power-law sing}: $\delta_0(t) \propto (T_*-t)^\Gamma\,,$ with $T_* < \infty.$
\begin{cor}
\label{cor:beta}
Under the hypotheses of Corollary \ref{cor:finite-time-sing}, the solution of the 3D Euler equations has a finite-time singularity at time $T_*<\infty,$ of power-law type with exponent $\Gamma,$ only if
\begin{equation*}
\Gamma \geq \frac{2}{6-n_-}, 
\end{equation*}
for all $n_-$ in $(-\infty, {\displaystyle \liminf_{t \to T_*} n_0(t)}] \cap (-\infty,6).$
\end{cor}
\begin{proof}[Proof] The proof follows directly from Corollary \ref{cor:finite-time-sing}. \end{proof}

\section{Analysis of analyticity-strip width in terms of BKM theorem}\label{Sec:NewAnal}

\subsection{Quality of bounds}

Several bounds were used in Section \ref{Sec:As_Bkm}.
We now proceed to test their sharpness, when they are applied to the numerical data of Section \ref{Sec:Numerics_Classical}.
\begin{figure}[htbp]
\begin{center}
\includegraphics[height=6.5cm]{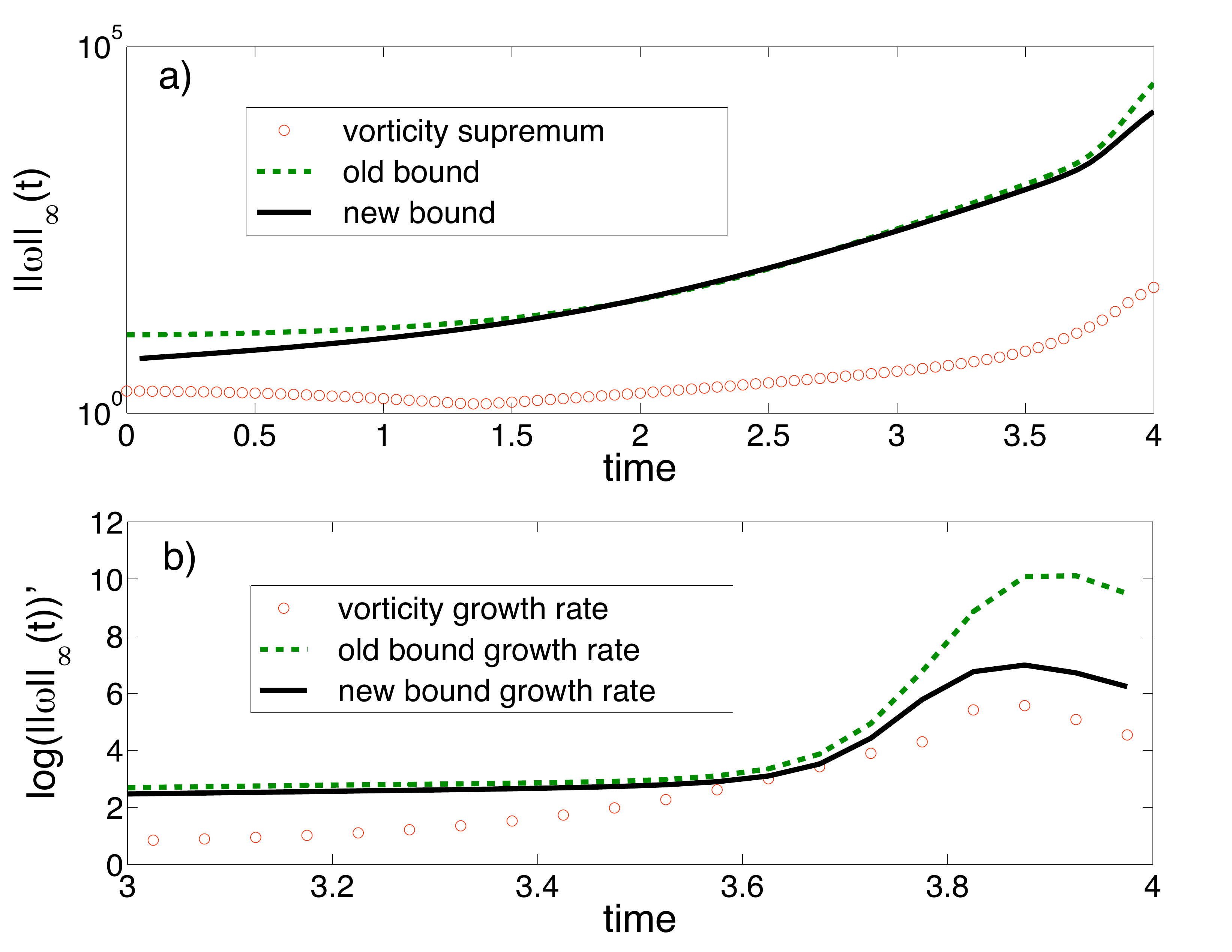}
\caption{(Color online)  Comparison of the bounds for the Taylor-Green flow at resolution $4096^3.$  a) Lin-Log plot, ``old bound'' is the RHS of the inequality (\ref{eq:ineq_old}), taking $\epsilon=0.1$ and $C_{\epsilon} = 3.9$ (see text). ``New bound'' is the RHS of the sharp inequality (\ref{eq:ineq_1}). b) Interpolated time derivative of the logarithms of a), for a time range localized near the change of trend. Same parameters as in a). \label{Fig:Figbounds}}
\end{center}
\end{figure}
Fig. \ref{Fig:Figbounds} shows a comparison of the new inequality (\ref{eq:ineq_1}), and the old inequality (\ref{eq:ineq_old}) taking $\epsilon=0.1$ 
with $C_{\epsilon} = 3.9$. Note that he value of $C_{\epsilon}$ (see Eq. \eqref{eq:Cepsilon}) can be estimated by the integral $\sqrt{\int_{\sqrt{3}}^\infty \pi k^2 k^{-3-2\epsilon} dk}=\sqrt{\pi3^{-\epsilon}/2 \epsilon}$ yielding $C_{\epsilon}\sim 3.75$ at $\epsilon=.1$. A more careful computation of the discrete sum gives $C_{\epsilon} \gtrsim 3.9$, the value used to generate Fig. \ref{Fig:Figbounds}.

The data in Fig. \ref{Fig:Figbounds}(a) displays two important facts: (i) The new bound is sharper than the old bound throughout the computation, particularly at the reliable end of the simulation, $t \gtrapprox 3.7,$ when the three curves show a change of trend and the old bound diverges at a faster rate than the new bound (see also Fig. \ref{Fig:Figbounds}(b)). (ii) Both old and new bounds are not too bad at the beginning of the computation ($t=0$), with an initial ratio of $5:2$ between the new bound and the vorticity supremum norm. Subsequently, the bounds become increasingly less sharp, and the new bound attains a ratio $165:1$ with the vorticity supremum norm at $t=4$. However the slope of the new bound's curve is comparable to the slope of the vorticity-supremum-norm curve.

In order to make a more quantitative comparison of the slopes, Fig. \ref{Fig:Figbounds}(b) shows the logarithmic rates of growth for old bound, new bound and vorticity supremum norm. In that order, these rates satisfy the ratios $7:5:4$ at the resolved time $t \approx 3.85$.

\subsection{Analysis of $\delta$ in terms of BKM}

We now proceed to see if the accelerated decay observed in the decrement $\delta(t)$ and quantified in Fig. \ref{Fig:Fit_Evolution}(d) can correspond to a power-law. To wit, we use the same local $3$-point method than that described in Section \ref{subsec:fit_methods_omegas} (see Eqs. \ref{eq:localfit}, \ref{eq:gamma}  and \ref{eq:Tc}).
The behavior of $g(t)$ is presented in Fig. \ref{Fig:Delta1} and the corresponding $T_*(t)$ and $\Gamma(t)$ are presented in Table \ref{Tab:table_dels}.
\begin{figure}[htbp]
\begin{center}
\includegraphics[height=6.5cm]{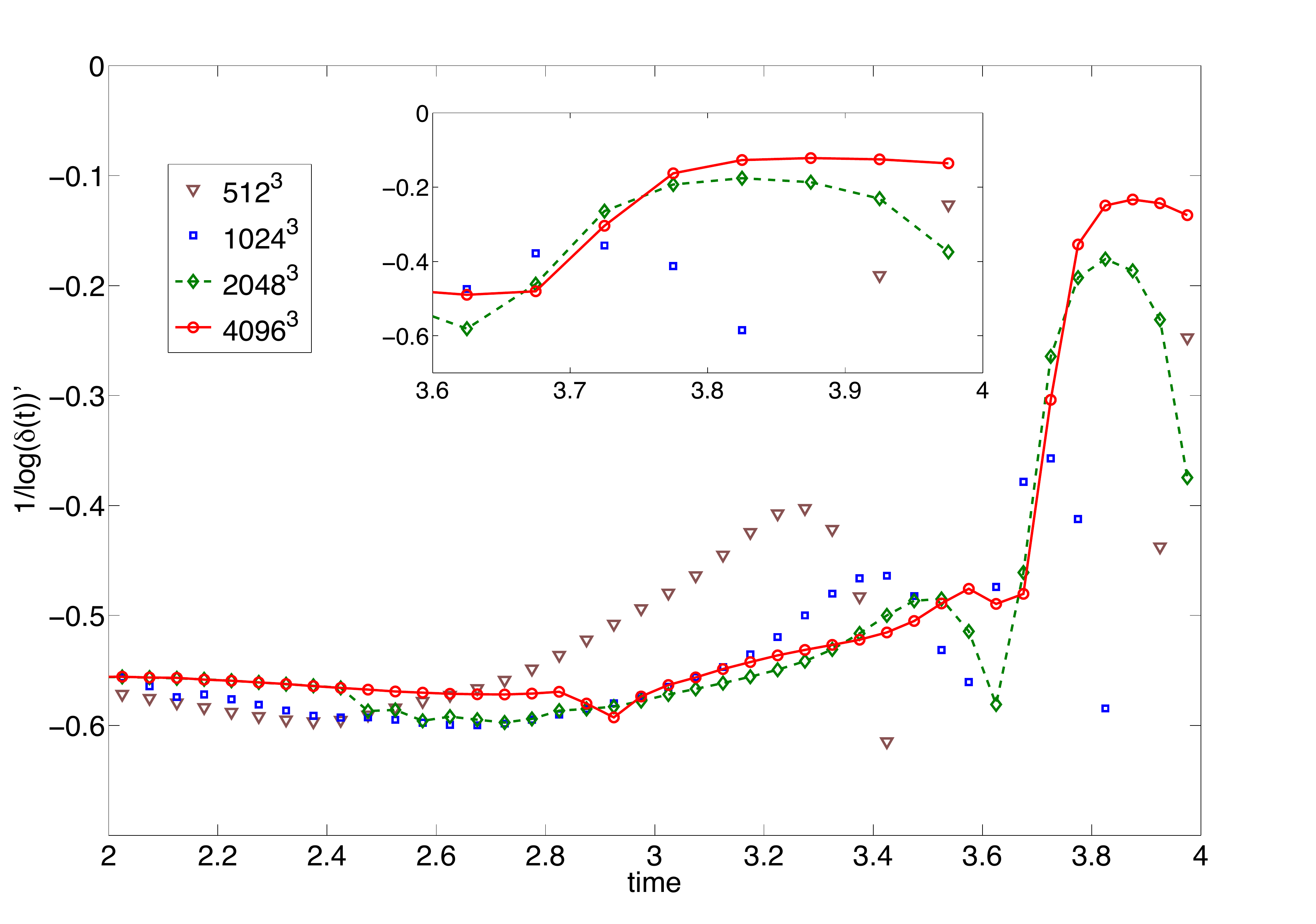}
\caption{(Color online) Temporal evolution of the inverse logarithmic derivative \eqref{eq:localfit} computed from the same values of $\delta$ as in Fig. \ref{Fig:Fit_Evolution}(d); $512^3$ (brown triangles), $1024^3$ (blue squares), $2048^3$ (green diamonds) and $4096^3$ (red circles).}\label{Fig:Delta1}
\end{center}
\end{figure}
\begin{table}[h]
\begin{tabular}{|c||c|c|c|c|}
	\hline
Time &  $\Gamma$   &   $\Gamma$    &   $T_*$      &   $T_*$ \\
         &$3 - k_{\max}$&  $103 - k_{\max}$   &$3  -  k_{\max}$&  $103 -  k_{\max}$\\
	\hline
3.7   & 0.283  & 0.383 & 3.81  &  3.83 \\
3.75 & 0.354 & 0.393 & 3.83  &  3.83 \\
3.8   & 1.41   &  1.36   & 4.00   &  3.97 \\
	\hline
\end{tabular}
\caption{Power-law fit parameters $\Gamma$ and $T_*$ (see Eq. \eqref{eq:sing}) for $\delta(t)$ determined at resolution $4096^3$ on full interval $3< k< k_{\max}$ (same as in Fig. \ref{Fig:Fit_Evolution} and Fig. \ref{Fig:Delta1}) and on subinterval $103< k< k_{\max}$ (see Table \ref{Tab:table_int}).}
\label{Tab:table_dels}
\end{table}

The results for exponent and predicted singular time of Table \ref{Tab:table_dels} have to be read carefully.
Because of the local $3$-point method used to derive them from the data in Table \ref{Tab:table_int}, they use the values of $\delta$ at $t=3.65, 3.7, 3.75, 3.8, 3.85$, the last one being marginally reliable (see Sec.\ref{Sec:NumEFits}).
In fact, they amount to linear $2$-point extrapolation of the data in Fig. \ref{Fig:Delta1} (see the inset): $T_*$ is the intersection of the straight line extrapolation with the time axis and $\Gamma$ is the inverse of the slope.
One can guess that there is room for a power-law type of behavior, with exponent $\Gamma \approx 0.4$ if we consider the data at $t=3.7, 3.75$ and $\Gamma \approx 1.4$ if we include the data at $t=3.8$.

We now use Corollary 11 (see Sec. \ref{Sec:As_Bkm}) to test if these estimates of power-law are consistent with the hypothesis of finite-time singularity. There, the product $\Gamma  (6-n_-)/2$ must be greater than or equal to one if finite-time singularity is to be expected. With the conservative estimate $n_- = 3.9$ obtained by inspection of Fig. \ref{Fig:Fit_Evolution}(b) (or equivalently using the values of $n$ in Table \ref{Tab:table_int}), we obtain that $\Gamma  (6-n_-)/2 < 1$ for the data at $t=3.7$ and $t=3.75$, but $\Gamma  (6-n_-)/2 >1$ for the data at $t=3.8$. These results are insensitive to the fit interval, see Table \ref{Tab:table_dels}.
Therefore, if the latest data is considered,  Corollary 11 cannot be used to negate the validity of the hypothesis of finite-time singularity.
However, there is no sign that the data values of  $\Gamma$ and $T_*$ in Table \ref{Tab:table_dels} are settling down into constants, corresponding to a simple power-law behavior.

Another piece of analysis consists of comparing the singular time predicted from the data for the decrement $\delta(t)$ with the singular time predicted from the direct data for the vorticity supremum norm. They seem both to be close to $T_* \approx 4$ (compare Table \ref{Tab:table_dels} to Table \ref{Tab:table_omegasup}).

In this context, we should perhaps mention Feynman's rule: ``Never trust the data point furthest to the right'', a comment attributed to Richard Feynman, saying basically that he would never trust the last points on an experimental graph, because if the people taking data could have gone beyond that, they would have.
Higher-resolution simulations are clearly needed to investigate whether the new regime is genuinely a power law and not simply a crossover to a faster exponential decay.

Our conclusion for this section is thus similar to that of Sec. \ref{subsec:fit_methods_omegas}: although our late-time reliable data for $\delta(t)$ shows $\Gamma  (6-n_-)/2 > 1$ and is therefore not inconsistent with our Corollary 11, clear power-law behavior of $\delta(t)$ is not achieved.

\section{Conclusions \label{Sec:Conclusion}}

In summary, we presented simulations of the Taylor-Green vortex with resolutions up to $4096^3$. We used the analyticity strip method to analyze the energy spectrum.  We found that, around $t\simeq 3.7$, a (well-resolved up to $t\simeq 3.85$) change of regime is taking place, leading to a faster decay of the width of the analyticity strip $\delta(t)$.
In the same time-interval, preliminary $3D$ visualizations displayed a collision of vortex sheets.
Applying the BKM criterium to the growth of the maximum of the vorticity on the time-interval $3.7<t<3.85$ we found that the occurrence of a singularity around $t\simeq 4$ was not ruled out but that higher-resolution simulations were needed to confirm a clear power-law behavior for $||\omega ||_\infty(t)$.

We introduced a new sharp bound for the supremum norm of the vorticity in terms of the energy spectrum. This bound allowed us to combine the BKM theorem with the analyticity-strip method and to show that a finite-time blowup can exist only if $\delta(t)$ vanishes sufficiently fast. Applying this new test to our highest-resolution numerical simulation we found that the behavior of  $\delta(t)$ is not inconsistent with a singularity. However, due to the rather short time interval on which $\delta(t)$ is both well-resolved and behaving as a power-law, higher-resolution studies are needed to investigate whether the new regime is genuinely a power law and not simply a crossover to a faster exponential decay.

Let us finally remark that our formal assumptions of Section \ref{subsec:main_results} are motivated and to some extent justified by the fact that, in systems that are known to lead to finite-time singularity, the analogous of the working hypothesis \eref{eq:fit_bound} is verified. For the analogy to apply, a version of the BKM theorem must be available. This is the case of the $1$-D inviscid Burgers equation for a real scalar field $u(x,t)$ defined on the torus:
\begin{equation*}
\frac{\partial u}{\partial t} + u \frac{\partial u}{\partial x} = 0 \quad \forall \, x \in [0,2 \pi], \,\forall \, t \in [0,T_*), 
\end{equation*}
which admits a BKM-type of theorem \cite{Bustamante2011}, with singularity time $T_*$ defined by
$\int^{T_*}\|u_x(\cdot,t)\|_{\infty}~dt = \infty$.

In the 1-D case, the analogous of our bound \eqref{eq:ineq_2_TYG} is
 \begin{equation*}
\left\| u_x(\cdot,t) \right\|_{\infty} \leq \tilde{c}\,{\displaystyle \sum_{k=1}^\infty k \sqrt{E(k,t)}} \,.  
 \end{equation*}

 Using the simple trigonometric initial data $u(x,0)=\sin(x)$, the energy spectrum can be expressed in terms of Bessel functions that admit simple asymptotic expansions. It is straightforward to show (see \cite{SulemSulemFrisch1983} for details) that, for $t<T_*=1$, one has the large-$k$ asymptotic expansion
\begin{equation*}
 E(k,t)\sim  \frac{1}{\pi  t^2 \sqrt{1-t^2} } k^{-3} e^{-2  \delta_S(t) k},
\end{equation*}
with
\begin{equation*}
\delta_S(t)= \log \left(\frac{\sqrt{1-t^2}+1}{t}\right)-\sqrt{1-t^2},
\end{equation*}
while, at $t=T_*=1$,
\begin{equation*}
E(k,1)\sim \frac{2\ 6^{2/3}}{ \Gamma \left(-\frac{1}{3}\right)^2} k^{-8/3}.
\end{equation*}
In fact, the $ k^{-8/3}$ 
power law appears already before $T_*$
(see the remark following Eq. (3-10) of reference \cite{SulemSulemFrisch1983}).

It is easy to check that the analytical solution admits, for all $k$ and for all $t$ sufficiently close to $T_*$, a working hypothesis \eref{eq:fit_bound} of the form 
\begin{equation*}
E(k,t) \leq M\,k^{-n_0}\,\exp(-2\,\delta_0(t)\,k),
\end{equation*}
with analytically-obtainable functions $n_0(t) = 8/3$ and $\delta_0(t) \propto (T_* - t)^\Gamma$ with $\Gamma = 3/2$. The analogous of Corollary \ref{cor:beta} gives the inequality
\begin{equation*}
\Gamma \geq \frac{2}{4 - n_0}\,,
\end{equation*}
which is saturated by the analytically-obtained exponents $n_0 = 8/3$, $\Gamma = 3/2$.

\section{Acknowledgements}
We acknowledge useful scientific discussions with Annick Pouquet, Uriel Frisch and Giorgio Krstulovic who also helped us with the visualizations of Fig. \ref{Fig:Vort_3D_Viz}.
The computations were carried out at IDRIS (CNRS). Support for this work was provided by UCD Seed Funding projects SF304 and SF564, and IRCSET Ulysses project ``Singularities in three-dimensional Euler equations: simulations and geometry''.

\appendix

\section{Extension to general periodic flows \label{Ap:Gen_per}}
Here we provide the generalization to non TG-symmetric periodic flows of the results presented in Section \ref{subsec:main_results}. Definition \ref{defn:spectrum} and the working hypothesis (Hypothesis \ref{hypo:working}) are modified slightly in the general case. Accordingly, the new bounds leading to Lemma \ref{lem:main} and Theorem \ref{thm:main} need to be modified slightly to accommodate the general case. The crucial derived relations between $\delta_0$ and $n_0$ in Lemma \ref{lem:strong} and Corollaries \ref{cor:finite-time-sing} and \ref{cor:beta} will apply directly to the general periodic case and will not be discussed.

The main technical difference is that the new bounds presented in Section \ref{subsec:main_results} apply for a flow with TG symmetries (see Section \ref{subsec:symm}) which imply that only modes with even-even-even and odd-odd-odd wavenumber components are populated. The general periodic case does not follow this restriction, which slightly modifies the bounds. We will assume, to simplify matters, that the so-called zero-mode of the velocity field is identically zero:
\begin{equation*}
 \widehat{\mbf{u}}(\mbf{0},t) = \mbf{0}\,, \quad \forall \,t \in [0,T),
\end{equation*}
 \\
Notice that all remaining wave numbers are populated. This means that all sums involving the scalar $k$ in equations  (\ref{eq:ineq_1}) and (\ref{eq:ineq_2_TYG}) will start effectively from $k=1.$

Also, because  modes with mixed even-odd wavenumber components are allowed, the definitions of $S_k$ in Lemma 2 and  constant $c$ in equation (\ref{eq:ineq_2_TYG}) must be replaced by more appropriate quantities. Therefore, the corresponding general periodic versions of Lemma \ref{lem:main} (equation (\ref{eq:ineq_1})) and practical bound (equation (\ref{eq:ineq_2_TYG})) are:\\

\noindent \textbf{Lemma \ref{lem:main}' (general periodic version of Lemma \ref{lem:main}).} \emph{Let $\mbf{u}(\mbf{x},t)$ be a velocity field with energy spectrum defined by equation \eref{eq:spectrum} and let $\bomega = \nabla \times \mbf{u}$ be its vorticity, defined on the periodicity domain $D=[0,2\,\pi]^3.$ Then the following inequality is verified for all times $t \in [0,T)$ when the sum in the RHS is defined, and independently of any evolution equation that $\mbf{u}$ might satisfy:}
\begin{eqnarray}
\label{eq:ineq_1_GP}
 \left\| \bomega(\cdot,t) \right\|_{\infty} &\leq& {\displaystyle \sum_{k=1}^\infty \,\sqrt{2\,k(k+1)\, E(k,t)\,S'_k}} \,,
\end{eqnarray}
\emph{where $S'_k \equiv \#\{\mbf{k} \in \mathbb{Z}^3 : k-1/2 < |\mbf{k}| < k+1/2\}$ is the number of lattice points in a spherical shell of width 1 and radius $k \in \mathbb{Z}_+$.}\\

\noindent \textbf{Practical bound, general case.} \begin{eqnarray}
\label{eq:ineq_2_GP}
 \left\| \bomega(\cdot,t) \right\|_{\infty} \leq c'\,{\displaystyle \sum_{k=1}^\infty k^2\sqrt{E(k,t)}} \,,
\end{eqnarray}
where $c'=6\sqrt{2}$.

We can easily check that the bounds for Taylor-Green, equations (\ref{eq:ineq_1}) and (\ref{eq:ineq_2_TYG}), are sharper (by a factor close to 2) to their respective general bounds, equations (\ref{eq:ineq_1_GP}) and (\ref{eq:ineq_2_GP}).

Finally, Theorem \ref{thm:main} is replaced by\\

\noindent \textbf{Theorem \ref{thm:main}'.} \emph{Let a solution of the 3D Euler equations satisfy the working hypothesis \eref{eq:fit_bound} with $k=1$ included. Then the maximal regularity time $T_*$ of the solution must satisfy}
\begin{equation*}
\int_0^{T_*} \Li{\frac{n_0(t)}{2}-2} {{\mathrm{e}}^{-\delta_0(t)}}\,dt = \infty.
\end{equation*}

\bibliographystyle{apsrev}

\end{document}